\def\notes{1}
\documentclass[11pt]{article}
\usepackage[letterpaper, portrait, margin=1in]{geometry}
\usepackage[hypertexnames=false,colorlinks=true,linkcolor=blue,citecolor=ForestGreen,pagebackref=true]{hyperref}
\usepackage{algorithm}
\usepackage[noend]{algpseudocode}
\usepackage{url}
\usepackage{thmtools,thm-restate}
\usepackage{amsmath,amssymb,amsthm}
\usepackage{thm-restate}
\usepackage[noabbrev,capitalise,nameinlink]{cleveref}
\usepackage{mathtools}
\usepackage{xspace}
\usepackage{verbatim}
\usepackage{mathrsfs}
\usepackage[usenames,dvipsnames,svgnames,table]{xcolor}
\usepackage{pgf}
\usepackage[dvipsnames]{xcolor}
\usepackage{todo}
\usepackage{tabularx}
\usepackage{enumitem}
\usepackage{derivative}
\usepackage{bm}
\usepackage{multirow}
\usepackage{diagbox}
\usepackage{nicematrix}
\usepackage{parskip}
\usepackage{adjustbox}
\usepackage{tikz}
\usepackage{tikz-3dplot}
\usepackage{transparent}
\usepackage{subcaption}
\usepackage[most]{tcolorbox}

\usepackage{dsfont}

\newcommand*\ie{i.\kern.1em e., }
\newcommand*\eg{e.\kern.1em g., }
\newcommand*\cf{c.\kern.1em f.\ }
\newcommand*\almev{a.\kern.1em e.\ }

\theoremstyle{plain}
\newtheorem{theorem}{Theorem}[section]
\newtheorem{lemma}[theorem]{Lemma}

\newtheorem{claim}[theorem]{Claim}
\newtheorem{corollary}[theorem]{Corollary}
\newtheorem{observation}[theorem]{Observation}

\theoremstyle{definition}
\newtheorem{definition}[theorem]{Definition}

\crefname{claim}{Claim}{Claims}
\crefname{fact}{Fact}{Facts}
\crefname{observation}{Observation}{Observations}

\tcbset{colframe=white,boxrule=0pt,frame hidden,parbox=false}
\tcbset{
    theoremstyle/.style={ colback=blue!7, },
    definitionstyle/.style={ colback=orange!10, },
    notestyle/.style={ colback=gray!8, }
}

\tcolorboxenvironment{restatable}{theoremstyle}
\tcolorboxenvironment{theorem}{theoremstyle}
\tcolorboxenvironment{lemma}{theoremstyle}
\tcolorboxenvironment{fact}{theoremstyle}
\tcolorboxenvironment{proposition}{theoremstyle}
\tcolorboxenvironment{corollary}{theoremstyle}

\tcolorboxenvironment{definition}{definitionstyle}
\tcolorboxenvironment{assumption}{definitionstyle}

\tcolorboxenvironment{conjecture}{notestyle}
\tcolorboxenvironment{observation}{notestyle}
\tcolorboxenvironment{remark}{notestyle}

\newcommand{\ignore}[1]{}

\newcommand{\dist}{\mathsf{dist}}
\newcommand{\eps}{\varepsilon}

\DeclareMathOperator{\Var}{Var}

\newcommand{\Ex}[1]{\bE \left[ #1 \right]}

\newcommand{\Pru}[2]{\underset{ #1 } {\mathrm{Pr}} \left[ #2 \right]}

\newcommand{\cA}{\ensuremath{\mathcal{A}}}

\newcommand{\cD}{\ensuremath{\mathcal{D}}}

\newcommand{\cH}{\ensuremath{\mathcal{H}}}

\newcommand{\cT}{\ensuremath{\mathcal{T}}}

\newcommand{\bE}{\ensuremath{\mathbb{E}}}

\newcommand{\bN}{\ensuremath{\mathbb{N}}}

\newcommand{\bZ}{\ensuremath{\mathbb{Z}}}

\newcommand{\qb}{\ensuremath n^{2/5}}

\usepackage{complexity}
\newlang{\NB}{NB}
\newlang{\DB}{DB}
\newlang{\DBW}{DBW}
\newlang{\DBWR}{DBWR}
\newlang{\DBWP}{DBWP}

\newcommand{\term}[1]{\ensuremath{[\![#1]\!]}}

\newcommand{\HS}{\textsf{HiddenString}\xspace}
\newcommand{\HSD}{\ensuremath $\textsf{\HS}^{\diamond}$\xspace}
\newcommand{\rev}[1]{#1^{\textsf{rev}}}

\newcommand{\Dyes}{\ensuremath \mathcal{D}^{+}}
\newcommand{\Dno}{\ensuremath \mathcal{D}^{-}}

\newcommand{\Bad}{\ensuremath \mathsf{BAD}}

\newcommand{\Tr}[1]{\ensuremath \cT \left[ #1 \right]}
\newcommand{\Tru}[2]{\cT_{ #1 } \left[ #2 \right]}

\usetikzlibrary{positioning, arrows, arrows.meta}
\usetikzlibrary{arrows.meta, decorations.pathreplacing, calc, shapes.geometric}
\usetikzlibrary{patterns}
\usetikzlibrary{backgrounds}
\ifnum\notes=1
    
    \newcommand{\srnote}[1]{{\color{blue}\footnote{{\color{blue} {\bf SR:} #1}}}}
    
    \newcommand{\dpnote}[1]{{\color{orange}\footnote{{\color{orange} {\bf DP:} #1}}}}
\else

    \newcommand{\srnote}[1]{}
    \newcommand{\dpnote}[1]{}
\fi

\usepackage{pifont}

\newtoggle{anonymous}
\togglefalse{anonymous}

\title{Query Complexity of Testing Structured Parenthesis Languages\footnote{The conference version of this paper will appear at SODA 2027.}}

\iftoggle{anonymous}{
    \author{}
}{
    \author{
    Tim Jackman\\ 
    Boston University
    \and 
    Diptaksho Palit\\ 
    Boston University
    \and
    Sofya Raskhodnikova\\ 
    Boston University
    } 
}

\date{\today}  
\begin{document}

\maketitle

\begin{abstract}
    We study the query complexity of testing membership in structured string
    languages, with a focus on Dyck languages and their natural generalizations.
    A tester receives query access to a word and must distinguish valid inputs from
    words that are far in Hamming distance, while inspecting only a sublinear number
    of positions.
    
    Our results sharpen the boundary between constant-query testability and polynomial query complexity.
    First, we prove an $\Omega(n^{2/5})$ lower bound for testing Dyck languages $D_m$ with any fixed number $m\ge 2$ of parenthesis types, improving the previous $\Omega(n^{1/5})$ lower bound of Fischer, Magniez, and Starikovskaya (SODA `18) and nearly matching their upper bound of $\Tilde{O}(n^{2/5+\delta})$ for all $\delta > 0$. Furthermore, we show that all nonadaptive algorithms for these problems require $\Omega(n^{1/2})$ queries. Our lower bounds use a Pólya-urn process in the construction of the
    hard distribution; to the best of our knowledge, this technique is new to
    property testing. 
    
    Second, we identify a broad class of weighted-parenthesis languages, which we
    call {\em excursion languages,} that remain constant-query testable. These languages
    encode bounded-step walks that stay nonnegative and return to zero. For every
    fixed excursion language, we give a nonadaptive tester with query complexity
    $O(1/\eps^2)$, and we prove this dependence on $\eps$ is
    optimal, even for adaptive algorithms. As a special case, we obtain the tight $\Theta(1/\varepsilon^2)$ query
    complexity of testing $D_1$, improving the previous
    $O(\log(1/\varepsilon)/\varepsilon^2)$ upper bound and giving the first matching
    two-sided-error lower bound.
    
    Third, we construct a simple hard language, Hidden String, that is generated by
    a deterministic linear grammar but nevertheless requires $\Omega(n^{2/5})$
    adaptive queries and $\Omega(n^{1/2})$ nonadaptive queries to test. This shows
    that polynomial query complexity appears even for highly restricted string
    languages.
\end{abstract}

\thispagestyle{empty}
\setcounter{page}{0}


\thispagestyle{empty}
\setcounter{page}{0}
\newpage
{
\setcounter{tocdepth}{2}
\tableofcontents
}
\thispagestyle{empty}
\setcounter{page}{0}
\newpage
\setcounter{page}{1}

\section{Introduction}
\label{sec:intro}

We study a basic question in sublinear algorithms: when can membership in a context-free language (CFL) be tested by inspecting only a small number of positions of the input word? CFLs include some of the most well-studied structured string languages, such as Dyck languages of balanced parentheses, and provide a natural setting for understanding the boundary between constant-query and polynomial-query testability.
We investigate this question through the lens of property testing \cite{RubinfeldS96, GoldreichGR98}: given query access to a word $w$, a tester for a language $L$ must distinguish the case $w \in L$ from the case that $w$ is $\eps$-far from every word in $L$, while reading only a small number of symbols of $w$.
One of the foundational results in property testing, due to Alon, Krivelevich, Newman, and Szegedy~\cite{AlonKNS00}, showed that every regular language is testable with a number of queries depending only on the distance parameter $\eps$, and independent of the input length.
This result has since been refined in a line of subsequent work, with Bathie, Fijalkow, and Mascle~\cite{BathieFM25} recently giving a trichotomy for the possible query complexities of regular languages. 

CFLs, however, exhibit a much richer picture.
Already \cite{AlonKNS00} showed that the Dyck language $D_1$ of balanced parentheses of one type has a constant-query tester, whereas the closely related language $D_2$ of balanced parentheses with two types does not. Subsequent work explored testability of several subclasses of CFLs---including parenthesis languages \cite{ParnasRR03,FischerMS18}, counter
automaton languages \cite{LachishNS08,GoldhirshV13}, and extensions to other structured
subclasses in related models~\cite{FischerMR10,FrancoisMRS16}---revealing a surprising mixture of constant-query and polynomial-query phenomena. Despite this progress, the structural reasons behind this mixture remain poorly understood: it is unclear which low-complexity CFLs inherit the strong testability guarantees of regular languages and which require polynomially many queries.
Even for the Dyck languages, the most extensively studied nonregular CFLs in this line of work, the landscape of known query-complexity bounds had remained unchanged since the work of Fischer, Magniez, and Starikovskaya~\cite{FischerMS18}.

We sharpen this picture both technically and structurally. Technically, we give the first improvement since \cite{FischerMS18} to the lower bounds for testing Dyck languages with multiple parenthesis types: for every fixed $m \ge 2$, we improve the query lower bound for $D_m$, the language of balanced parentheses with $m$ types, from $\Omega(n^{1/5})$ to $\Omega(n^{2/5})$ (nearly matching their upper bound of  $\Tilde{O}(n^{2/5+\delta})$ for all $\delta > 0$) and show, in addition, that nonadaptive algorithms (i.e., algorithms that provide all queries in advance, before receiving any answers) require $\Omega(n^{1/2})$ queries. 
Our lower bounds build on the hidden-alignment constructions of
\cite{ParnasRR03,FischerMS18}, but introduce new techniques. We use a Pólya-urn process in the constructions of the hard distributions; to the best of our knowledge this is the first use of such a process in property testing. Moreover, for adaptive
testers, we use an amortized analysis of the tester's progress, avoiding the
round-by-round worst-case losses in previous analyses. 
Structurally, we expand the understanding of which low-complexity CFLs are easy or hard to test. On the positive side, we identify a broad class of weighted-parenthesis languages, which we call \textit{excursion languages}, that retain the constant-query testability of $D_1$; our bounds are optimal as a function of $\eps$ and, in particular, pin down the query complexity of $D_1$ as $\Theta(1/\eps^2)$. On the negative side, we show that polynomial query complexity already appears at the bottom of a hierarchy of deterministic linear languages, demonstrating that linearity of grammars alone does not explain constant-query testability.

\subsection{Our results}\label{sec:results}

Dyck languages are among the most well-studied nonregular CFLs.
For $m\in\bN$, the {\em Dyck-$m$ language $D_m$} consists of all well-balanced strings of parentheses with $m$ types.
The query complexity of testing
Dyck languages has been studied since the early work on property testing.
In addition to giving a constant-query tester for $D_1$,
\cite{AlonKNS00} showed that $D_2$ is not constant-query testable. The
quantitative study of Dyck languages with multiple types of parentheses
was initiated by \cite{ParnasRR03}, who gave an adaptive $\widetilde{O}(n^{2/3}/\eps^3)$-query tester and an adaptive $\widetilde\Omega(n^{1/11})$-query lower bound for testing $D_m$ for all $m \geq 2$. These bounds were later improved in
\cite{FischerMS18}. In particular, \cite{FischerMS18} gave, for every
fixed $m\ge 2$ and every constant $\alpha>0$, an
$\widetilde O(n^{2/5+\alpha})$-query tester for $D_m$, and proved an
$\Omega(n^{1/5})$ lower bound. 
Our first result improves this lower bound to $\Omega(\qb)$.

\begin{theorem}[Dyck language lower bound for $m\geq 2$]\label{thm:dyck-lb}
    For every $m\ge 2$, every $\eps$-tester for $D_m$, for sufficiently
    small constant $\eps>0$, requires $\Omega(\qb)$ queries on inputs of
    length $n$.
\end{theorem}

\Cref{thm:dyck-lb} strengthens the known separation between $D_1$ and
Dyck languages with two or more types of parentheses: $D_1$ has
constant-query testers, whereas testing $D_m$ for every fixed $m\ge 2$
requires $\Omega(\qb)$ queries.
Another interesting contrast between $D_1$ and Dyck languages $D_m$ with $m\geq 2$ is that $D_1$ is testable nonadaptively, whereas the current algorithms for $m\geq 2$, both from \cite{ParnasRR03} and \cite{FischerMS18}, are adaptive. This raises the question of whether adaptivity is necessary to achieve better query complexity for languages with two or more types of parentheses. 
We answer this affirmatively.

\begin{theorem}[Nonadaptive Dyck language lower bound for $m\geq 2$]\label{thm:dyck-lb-nonadaptive}
For every $m\ge 2$, every nonadaptive $\eps$-tester for $D_m$, for sufficiently
small constant $\eps>0$, requires $\Omega(n^{1/2})$ queries on inputs of
length $n$.
\end{theorem}

The proofs of the two theorems require a different hard distribution than in previous lower bounds.
Positive and negative instances in existing constructions can be distinguished by nonadaptive algorithms making roughly $\widetilde O(n^{1/4})$ queries, so improving the exponent necessitates new hard instances.
We construct new  hard distributions and establish a polynomial separation between adaptive and nonadaptive testing for Dyck languages $D_m$.
\Cref{thm:dyck-lb-nonadaptive} also provides a new kind of separation between $D_1$ and
Dyck languages with two or more types of parentheses: $D_1$ has nonadaptive constant-query testers, whereas every nonadaptive tester for $D_m$ for $m\geq 2$ requires $\Omega(n^{1/2})$ queries.

\paragraph{Excursion languages.} \cref{thm:dyck-lb,thm:dyck-lb-nonadaptive} raise the question of which other natural extensions of $D_1$ remain constant-query testable by simple, nonadaptive algorithms. Our second direction answers this question for a different extension, in which parentheses are allowed to have bounded integer weights.
Instead of having only one opening symbol and one closing symbol, each symbol contributes a bounded integer number of parentheses to the word.
For fixed $\ell,r\in\mathbb N$, consider the alphabet
$
    \Sigma_{\ell,r}=\{-\ell,-\ell+1,\ldots,-1,0,1,\ldots,r\}.
$
A positive symbol $\sigma\in\Sigma_{\ell,r}$ is interpreted as $\sigma$ open parentheses, whereas a negative symbol $\sigma$ represents $|\sigma|$ closing parentheses; zeros are interpreted as other expressions that don't need to be balanced. The language $L_{\ell,r}$ is defined as the set of all strings of balanced parentheses over the alphabet $\Sigma_{\ell,r}$ (see \Cref{def:excursion-language}). 

We call $L_{\ell,r}$ an {\em excursion} language because words over $\Sigma_{\ell,r}$ can be viewed as walks on the integer line, where a symbol $\sigma$ changes the height by $\sigma$. The language $L_{\ell,r}$ consists
of the walks that start at height $0$, stay nonnegative, and return to
height $0$. Such walks, called excursions,  are classical objects of study in enumerative combinatorics. Special cases include Motzkin paths, with step size in
$\{-1,0,1\}$, and bounded \L{}ukasiewicz excursions with step size in
$\{-1,0,1,\ldots,r\}$.  
Excursion languages $L_{\ell,r}$ form a simple family of deterministic one-counter
languages: the counter stores the current height of the walk.

We show that all excursion languages (for fixed $\ell$ and $r$) are testable with $O(1/\eps^2)$ queries by nonadaptive algorithms and, moreover, that this bound is tight, even for adaptive algorithms.

\begin{theorem}[Testability of excursion languages]\label{thm:excursion}
Let $\eps\in(0,1)$. For all fixed $\ell,r\in\bN$, the query complexity of testing the language $L_{\ell,r}$ with distance parameter $\eps$ is $\Theta(1/\eps^2)$. This holds for both adaptive and nonadaptive algorithms.
\end{theorem}

This theorem extends the testability result of \cite{AlonKNS00} for the Dyck-1 language, $D_1$, to a new subclass of CFLs. Each excursion language $L_{\ell,r}$ is recognized by the following deterministic one-counter pushdown automaton which uses its counter to track the sum of the current prefix. On reading a symbol $a>0$, the automaton increments the counter by $a$. On reading $a=0$, it leaves the counter unchanged. On reading $a<0$, it
attempts to decrement the counter by $|a|$, rejecting if the counter reaches zero before it is done decrementing. The automaton accepts iff the input ends with counter value $0$.

Our positive result cannot be extended to all deterministic single-counter automata languages, as Lachish, Newman, and Shapira~\cite{LachishNS08}  showed that such automata can recognize languages requiring $\Omega(\log \log n)$ queries to test.
Goldhirsh and Viderman \cite{GoldhirshV13} also investigated the question of classifying which languages  recognized by deterministic counter automata admit constant-query testers. Their main result is that every language recognized by a {\em weak} counter automaton is testable with a constant number of queries, where weakness means that transitions available when the counter is zero are a subset of those available when the counter is positive. This result does not subsume ours:
weak automata cannot require the counter to be zero in order to accept; they can recognize prefixes of excursion languages, but not the excursion languages themselves.

\paragraph{The Dyck-1 language.} Our techniques for excursion languages also apply to the Dyck-1 language $D_1$, pinning down its query complexity to $\Theta(1/\eps^2)$. This gives a logarithmic improvement over the bound of $O(\log(1/\eps)/\eps^2)$ by Alon et al.~\cite{AlonKNS00} and shows that our algorithm for $D_1$ is optimal. Alon et al.\ proved an $\Omega(n)$ lower bound on the query complexity of every tester for $D_1$ that has one-sided error. However, no two-sided error lower bound for $D_1$ was previously known.

\paragraph{Languages generated by linear grammars.}
Our third direction shows that polynomial query complexity arises even
within a very restricted subclass of the {\em linear languages,} i.e., languages generated by linear grammars. Linear grammars only have production rules that contain at most one nonterminal on the right-hand side. 
Linear languages lie strictly between regular languages and CFLs. As explained in \cref{sec:hierarchy}, they can equivalently be recognized by nondeterministic biautomata, which read the input from both ends
\cite{HolzerJ14}. Imposing determinism and further restrictions yields a
hierarchy of subclasses of the linear languages \cite{JiraskovaK22a} (see \cref{fig:NB-hierarchy}).

We exhibit a language, which we call \HSD, in the bottom
class of this hierarchy, the Nasu--Honda deterministic linear languages
\cite{NasuH69}, and prove that it nevertheless requires polynomially many
queries to test. \HSD is the language over the alphabet $\{ 0^H, 1^H, \star, \diamond, 0, 1 \}$ generated by the linear grammar $S \to \star S \mid 0^HS0 \mid 1^HS1 \mid \diamond$.
Each word in \HSD has the form $u\diamond v$, where $v$ is a binary string and
$u$ is obtained from the reverse of $v$ by marking its symbols with a superscript $H$ and inserting $\star$'s in arbitrary positions. Thus, the non-$\star$ symbols of $u$ form the marked reverse of $v$, which is ``hidden'' in $u$. This language is inspired by {\sf TrueStringEqui\-valence} used by \cite{FischerMS18} to give a lower bound for languages $D_m$,
which we had to modify to get a Nasu--Honda deterministic linear language. 
Although the class of such languages is very restricted and incomparable with regular languages, it is powerful enough to contain a hard-to-test language.

\begin{theorem}[Hidden String lower bound]\label{thm:HSD-lb}
   For every sufficiently small constant $\eps > 0$, for inputs of length $n$, every adaptive $\eps$-tester for \HSD requires
   $\Omega(\qb)$ queries and every nonadaptive $\eps$-tester for \HSD requires $\Omega(n^{1/2})$ queries.
   Moreover, \HSD is a Nasu--Honda deterministic linear language.
\end{theorem}

Thus, even the most restricted level of the linear-language hierarchy \cite{JiraskovaK22a}
already contains a language whose membership is not constant-query
testable.

\subsection{Our techniques}\label{sec:techniques}
Our main technical contributions are the lower bounds for \HSD (from which we derive the same lower bounds for the Dyck-$m$ languages with $m\geq 2$) and the algorithm for the excursion languages and its analysis (from which we derive a tight upper bound for the Dyck-1 language). 

\paragraph{The lower bounds for \HSD.}

Our lower bounds build on the framework introduced by \cite{ParnasRR03} and refined by  \cite{FischerMS18}, but both the adaptive and nonadaptive analyses require new ideas.
In the previous constructions, both strings (in our case, separated by $\diamond$) contain dummy symbols (in our case, stars), which hide the alignment between the corresponding $\{0,1\}$ symbols in the two parts of the string.
While we modify the language to make it simpler in the formal-language sense---our variant is tailored to lie at the bottom level of the linear-language hierarchy---our language retains the features that make it difficult for testing.

In the construction of the hard instances for both our lower bounds, we keep the idea of inserting stars in random positions to make checking the alignment of the left and right parts of the string more difficult.
Our hard distributions are over input strings of length $3n+1$ that are concatenations of two parts separated by $\diamond$.
The hidden (left) part of our input has length $2n$ and is organized into symmetric pairs: for each $i\in[n]$, exactly one of the two positions
$i$ and $2n+1-i$ contains a $\star$.
For both our lower bounds, we choose the subsequence of $\{0,1\}$ symbols in the left (hidden) part of the string independently and uniformly at random.
For positive instances, we make the right (clear) part of the string equal to the reverse of $u$;
for negative instances, we resample the bits in the middle third of that string as illustrated in \cref{fig:filter-example}.
The main new ideas in our two lower bounds are in how we choose where to place $\star$ symbols in the hidden string.
In prior work \cite{ParnasRR03,FischerMS18}, the location of $\star$ for each pair is chosen independently and uniformly at random.

An algorithm has to catch a pair of corresponding symbols (i.e., in ``matching'' positions) in the middle third from both parts of the string to observe any difference between the two distributions.

\begin{figure}[h]
    \centering
    \begin{tikzpicture}[
    cell/.style={draw, minimum width=0.75cm, minimum height=0.75cm},
    arr/.style={{Triangle[length=1mm,width=2mm]}-{Triangle[length=1mm,width=2mm]}, thick},
    arr2/.style={-{Triangle[length=1mm,width=2mm]}, thick},
    font=\normalsize
]
\def\n{10}
\foreach \i in {1,...,\n} {
    \node[cell] (c\i) at (\i*0.75, 0) {};
}
\node at (c1) {{\Large$\star$}};
\node at (c2) {{\Large$\star$}};
\node at (c4) {$\ldots$};
\node at (c5) {{\Large$\star$}};
\node at (c7) {$\ldots$};
\node at (c8) {{\Large$\star$}};
\def\bh{0.35} 

\draw[arr] (c10.north) -- ++(0,0.55) -| (c1.north);
\draw[arr] (c9.north) -- ++(0,0.425) -| (c2.north);
\draw[arr] (c8.north)  -- ++(0,0.3) -| (c3.north);
\draw[arr] (c6.north)  -- ++(0,0.175) -| (c5.north);
\node[minimum width=0.75cm, minimum height=0.75cm, anchor=west] 
    (dia) at ($(c10.east)$) {{\Large$\diamond$}};
\node[draw, minimum width=1.8cm, minimum height=0.75cm, anchor=west]
    (r1) at ($(dia.east)$) {};
\node[draw, minimum width=1.8cm, minimum height=0.75cm, anchor=west,
    pattern=north east lines, pattern color=blue!30]
    (r2) at (r1.east) {};
\node[draw, minimum width=1.8cm, minimum height=0.75cm, anchor=west]
    (r3) at (r2.east) {};
\node[above=0.5cm] at (r2.north) (rlabel) {Resampled for $\Dno$};
\draw[arr2] (rlabel.south) -- (r2.north);

\draw[{Latex[length=1.5mm]}-{Latex[length=1.5mm]}]
    ($(r2.west) + (0.1,0)$) -- ($(r2.east) + (-0.1,0)$)
    node[midway, fill=white, inner sep=1pt, font=\small] {$n/3$};


\def\ylevel{-0.6}

\draw[{Latex[length=1.5mm]}-{Latex[length=1.5mm]}]
    (r1.south west |- 0,\ylevel) -- (r3.south east |- 0,\ylevel)
    node[midway, fill=white, inner sep=1pt] {$n$};

\draw[{Latex[length=1.5mm]}-{Latex[length=1.5mm]}]
    (c1.south west |- 0,\ylevel) -- (c5.south east |- 0,\ylevel)
    node[midway, fill=white, inner sep=1pt] {$n$};

\draw[{Latex[length=1.5mm]}-{Latex[length=1.5mm]}]
    (c6.south west |- 0,\ylevel) -- (c10.south east |- 0,\ylevel)
    node[midway, fill=white, inner sep=1pt] {$n$};
\end{tikzpicture}
    \caption{
    Hard instances for our lower bounds. 
    The $2n$ symbols on the left form the \emph{hidden} string; the $n$ bits on the right form the \emph{clear} string.
    For each $i\in[n]$, exactly one of the positions $i$ and $2n+1-i$ of the hidden string contains a $\star$; Figure~\ref{fig:filters} shows the two filter distributions that decide which one.
    For $\mathcal{D}^-$, the middle third of the clear string is resampled.
    }\label{fig:filter-example}
\end{figure}

\paragraph{Nonadaptive lower bound.}
We call the locations of the $\star$s in the first half of the hidden string the \emph{filter} $X \in \{0,1\}^n$:
for each $i\in[n]$, $X_i=1$ if position~$i$ of the hidden string holds a bit, and $X_i=0$ if it holds a $\star$ (in which case the bit is at position $2n+1-i$).
With independent uniform filter bits, as in prior work, the true index $X_1+\dots+X_j$ of a hidden position $j \in (n/3,n]$ is concentrated in a window of width $O(\sqrt{n})$ around $j/2$.
So a suitably aligned pair of queries forms a matching pair with probability $\Theta(1/\sqrt{n})$, and one can in fact find such a pair with $\Tilde{O} (n^{1/4})$ queries.
Our filter instead has a random density: we pick $p$ uniformly from $[0,1]$ and then set each $X_i$ to $1$ independently with
probability~$p$.
Equivalently, we pick the weight of the filter uniformly from $\{0,1,\dots,n\}$ and then a uniformly random set of positions of that size;
equivalently again, $X$ is generated by a P\'olya urn that starts with one $0$-ball and one $1$-ball.
The urn description shows that the weight of every prefix $X_1 \dots X_s$ is uniform over $\{0,1,\dots,s\}$.
Hence each pair of a queried hidden position in $(n/3,2n-n/3]$ and a queried clear position forms a matching pair with probability less than $3/n$, and a union bound over the $o(n)$ pairs of queries shows that, with probability $1-o(1)$, a nonadaptive tester making $o(n^{1/2})$ queries never queries both positions of a matching pair involving the middle third of the clear string.

\paragraph{Adaptive lower bound.}
For adaptive testers, the analysis, described below, reveals to the tester growing parts of the hidden string at its two ends.
With a single global density, these parts would essentially determine the density: conditioned on them, the remaining filter is
again a P\'olya urn biased toward the observed density, so the next stretch of the filter behaves essentially like independent coins, as in previous constructions.
We instead divide the first half of the hidden string into $n^{2/5}$ blocks of length $b = n^{3/5}$ and draw the filter of each block independently from our nonadaptive filter distribution, on $b$ positions instead of~$n$.
Every block then contributes a fresh, uniformly random number of bits:
if there are $\ell$ unrevealed blocks between the new query and the revealed part, the true index of the queried position involves a sum of $\Theta(\ell)$ independent uniform variables over $\{0,1,\dots,b\}$, which takes any fixed value with probability $O(1/(b\sqrt{\ell}))$, compared with $\Theta(1/\sqrt{b\ell})$ for independent filter bits.

We measure the progress of a $q$-query tester by a \emph{frontier} $h$, which signifies the ``farthest'' matching pair witnessed by the tester so far.
We reveal to the tester every symbol within distance $h$ of either end of the hidden
string. 
The tester can observe a difference between the two distributions only after the frontier passes $n/3$.
As detailed in \cite{ParnasRR03}, one of the obstacles in analyzing this type of construction is that even pairs of queries that do not match carry information, and this must be conditioned on.
More specifically, for queries that are far from the ``farthest'' matching pair witnessed, such a mismatch occurs with high probability and can thus be safely ignored, but for queries that are close to the ``farthest'' matching pair witnessed, a mismatch might not occur with high probability and subsequently needs careful conditioning.
Previous work handled this by revealing the matched positions of all queries close to the furthest matching pair found so
far.
This is too costly for us, so we introduce a new reveal condition.
Our frontier moves past a group of queries near it only when the advance is proportional to the size of the group:
whenever the $k$-th closest query beyond the frontier lies within $2k$ blocks of it, the frontier moves past that query. Such a \emph{free advance} costs at most $2b$ per query, and afterwards the distance between the frontier and each remaining query is large enough that conditioning on them at most doubles the relevant probabilities.
Moreover, a query $\ell$ blocks beyond the frontier completes a matching pair with probability only $O(q/(b\sqrt{\ell}))$, in which case the frontier \emph{directly advances} by $b\ell\le n$ to move past it.
Finally, rather than controlling the progress of the tester round by round, we instead bound the expected total progress over the whole execution, which avoids repeatedly paying for worst-case behavior at every step.
This expected progress is $O(qb+q^2\sqrt{n/b})$ (the first term comes from free advance of $O(b)$ per query and the second from direct advance of $b \ell$ with probability $O(q/(b\sqrt{\ell}))$ per query).
With $b=n^{3/5}$, which balances the two terms, this is $o(n)$ whenever $q=o(n^{2/5})$, and Markov's inequality shows that the frontier is unlikely to pass $n/3$.

\begin{figure}[h]
    \centering
    \newcommand{\FilterStrip}[1]{%
        \foreach \b [count=\i] in {#1} {%
        \ifnum\b=1 \fill[black!25] (\i-1,0) rectangle (\i,1);
        \else \node[inner sep=0pt] at (\i-0.5,0.5) {\scriptsize$\star$};\fi
        \draw[black!55] (\i-1,0) rectangle (\i,1);
    }}%
    \newcommand{\FilterPath}[2]{%
        \gdef\FPsum{0}%
        \foreach \b [count=\i] in {#1} {%
        \pgfmathtruncatemacro{\FPnew}{\FPsum+\b}%
        \draw[very thick, line cap=round] (\i-1,{#2*\FPsum}) -- (\i,{#2*\FPnew});
        \xdef\FPsum{\FPnew}%
    }}%
    \begin{subfigure}[t]{0.48\textwidth}
        \centering
        \begin{tikzpicture}[x=0.37cm, y=0.37cm, font=\footnotesize]
            \FilterStrip{1,1,0,1,1,1,0,1,0,1,1,1,0,1,1,0,1,0,1,1}
            \draw[decorate, decoration={brace, amplitude=3pt}] (0,1.3) -- (20,1.3)
            node[midway, above=3pt] {weight $W$ uniform in $\{0,1,\dots,n\}$};
        \end{tikzpicture}
        \caption{Nonadaptive filter.}
        \label{fig:filter-nonadaptive}
    \end{subfigure}%
    \hfill
    \begin{subfigure}[t]{0.48\textwidth}
        \centering
        \begin{tikzpicture}[x=0.37cm, y=0.37cm, font=\footnotesize]
            \FilterStrip{1,0,1,1,1, 0,0,0,1,0, 1,1,1,1,1, 0,1,0,0,1}
            \foreach \x in {5,10,15} \draw[thick] (\x,0) -- (\x,1);
            \foreach \k in {1,...,4}
            \draw[decorate, decoration={brace, amplitude=2.5pt}]
            ({5*(\k-1)+0.15},1.3) -- ({5*\k-0.15},1.3) node[midway, above=2pt] {$W_{\k}$};
            \node[above] at (10,2.9) {weights $W_i$ uniform in $\{0,1,\dots,b\}$, independent};
        \end{tikzpicture}
        \caption{Adaptive filter.}
        \label{fig:filter-adaptive}
    \end{subfigure}
    \caption{The two filter distributions.
    Each strip shows the first half of the hidden string: shaded cells hold a bit ($X_i=1$), and the other cells hold a~$\star$ ($X_i=0$).}
    \label{fig:filters}
\end{figure}

\paragraph{The tester for excursion languages.} 
Our tester for excursion languages $L_{\ell,r}$ views the input word as a walk on the integer line. Membership means that the walk starts at $0$, never goes below $0$, and returns to $0$. The central statistic in the analysis is
\[
    \Delta(w)=s(w)-2m(w),
\]
where $s(w)$ is the final height and $m(w)$ is the minimum height reached
by the walk. This statistic captures the two ways in which a word can fail
to be an excursion: the walk may end at the wrong height, or it may dip
below $0$. It is inspired by techniques of \cite{BermanRaskhodnikovaYaroslavtsev14} for tolerant testing of monotonicity of Boolean functions.

The main structural ingredient in the analysis is a Repair Lemma (\cref{lem:repair}) showing that
\[
    \dist(w,L_{\ell,r}) \le \Delta(w).
\]
It implies that for words $w$ that are $\eps$-far from
$L_{\ell,r}$, our statistic is large: $\Delta(w)\ge \eps n$. 
The proof of the Repair Lemma is constructive. We repair all negative prefixes by modifying the steps that first reach new negative heights. Then, once the
walk is nonnegative, we repair the final height by modifying the last  upward crossings of positive heights.

The tester estimates $\Delta(w)$ by sampling $O(1/\eps^2)$ positions (for fixed $\ell$ and $r$),
computing $\Delta$ on the sampled walk and then rescaling. Soundness follows
because, for a far input, the sampled statistic is unlikely to
substantially underestimate $\Delta(w)$; it suffices to control the sampled
walk at the endpoint and at a point where the original walk reaches its
minimum. Completeness follows because, for a valid excursion, the expected
sampled walk is a scaled valid excursion, and the deviations from this
expectation are controlled by standard martingale and variance bounds.
Testing the Dyck-1 language, $D_1$, easily reduces to testing $L_{1,1}$, yielding the $O(1/\eps^2)$-query tester for $D_1$ that avoids the
logarithmic loss in the original tester of \cite{AlonKNS00}.

\subsection{Related work}\label{sec:prior-work}

Testability of languages has been extensively studied.

\paragraph{Regular languages.} Following the work of \cite{AlonKNS00}, property testing of regular languages has developed into a rich line of research. 
Subsequent work considered distance measures besides the Hamming distance: edit distance with moves \cite{MagniezR07,FischerMR10} and weighted edit distance. Under the latter measure, \cite{FrancoisMRS16} gave a $O(\log^2(1/\eps)/\eps)$-query tester and \cite{BathieS21} gave a tester with optimal query complexity $O(\log(1/\eps)/\eps)$. Recently, \cite{BathieFM25} established a trichotomy for testing with one-sided error under the Hamming distance: every regular language is testable with either $\Theta(\log(1/\eps)/\eps)$, $\Theta(1/\eps)$, or $0$ queries 
for sufficiently large input lengths.
Beyond the query model, \cite{GanardiHL16, GanardiHLS19, GanardiHLMS25}  studied testing for regular languages in the sliding window model.

\paragraph{Other extensions of testing of languages.}
Property testing has also been studied for richer classes of languages, including regular tree
languages~\cite{MagniezR07}, context-free and infinite regular languages
under edit distance with moves~\cite{FischerMR10}, languages recognized by
small-width branching programs~\cite{Newman02,FischerNS04}, and visibly
pushdown languages in the streaming model~\cite{FrancoisMRS16}.

\paragraph{Concurrent and independent work.}
Concurrently with 
this work, Strauss \cite{Strauss26} also showed that testing Residual-String Equality (the Truestring Equivalence problem of \cite{FischerMS18}), and hence $D_m$ for every $m \ge 2$, requires $\Omega(n^{2/5})$ queries, and $\Omega(\sqrt{n})$ queries nonadaptively.
She also gave a nonadaptive $O(n^{1/2+\delta})$-query tester for $D_m$, so \cref{thm:dyck-lb-nonadaptive} is nearly tight, and improved the $\varepsilon$-dependence of the tester of \cite{FischerMS18}.

\subsection{Open questions}\label{sec:open-questions}

Our results sharpen the boundary between constant-query 
and polynomial-query testability for several
natural subclasses of CFLs. They leave open the broader question of how hard CFL membership testing can be.


\paragraph{Testing arbitrary CFLs.}
The largest known lower bound for testing a CFL is the $\Omega(\sqrt n)$ lower bound of Alon et al.~\cite{AlonKNS00} for the language $\{v\rev v w \rev w : v,w\in\{0,1\}^*\}$. It remains open whether
this is the right worst-case behavior for CFLs. Is every CFL testable with $O(\sqrt n)\cdot \operatorname{poly}(1/\eps)$ queries, or is there a CFL whose testing complexity is polynomially larger, perhaps even linear in~$n$?

\paragraph{Testability of subclasses of CFLs.}
Another direction is to understand the query complexity of natural subclasses of CFLs as a function of $n$.
Regular languages, weak one-counter languages, and, as we show, excursion languages are constant-query testable.
In contrast, our results establish that even very restricted linear languages can require polynomially many queries.
Can one characterize which natural subclasses of CFLs have constant, polylogarithmic, polynomial, or near-linear query complexity in $n$?

\subsection{Organization}\label{sec:organization}

In \Cref{sec:prelim}, we review the background on formal languages and property testing.
In \cref{sec:lb}, we present our adaptive $\Omega(\qb)$ and nonadaptive $\Omega(n^{1/2})$ lower bounds for testing Dyck languages with at least two types of parentheses and \HSD, proving \cref{thm:dyck-lb,thm:dyck-lb-nonadaptive,thm:HSD-lb}.
In \Cref{sec:excursion-languages}, we give our tester for excursion languages, proving the upper bound part of \cref{thm:excursion} (stated separately in \cref{thm:excursion-up}) and derive the optimal $O(1/\eps^2)$-query tester  for Dyck-1. In \cref{sec:optimality}, we present the matching $\Omega(1/\eps^2)$ lower bounds for Dyck-1 and excursion languages, proving the optimality of our testers and completing the proof of \cref{thm:excursion}.
\section{Preliminaries}\label{sec:prelim}

An alphabet $\Sigma$ is a finite set of symbols; a word $w$ over $\Sigma$ is a finite string of symbols from $\Sigma$. The symbol at the $i^{\text{th}}$ index of $w$ is denoted by $w[i]$. The length of a word $w$ is denoted by $|w|$ and the empty string, $\lambda$, is the unique word of length $0$. The set of all words of length $n$ over $\Sigma$ is $\Sigma^n$ and $\Sigma^* = \bigcup_{i \in \bN \cup \{0\}} \Sigma^i$. A language $L$ over an alphabet $\Sigma$ is a subset of~$\Sigma^*$. We use $\circ$ to denote concatenation and $\term{a}^n$, where $a \in \Sigma$ and $n \in \bN$, to denote the string consisting of $n$ copies of $a$. 

The \emph{distance} between two words $w$ and $w'$ in $\Sigma^n$, written as $\dist(w,w')$, is the \emph{Hamming distance}, i.e., the number of indices on which they differ. The distance between a word $w$ and a language $L$, denoted $\dist(w,L)$, is the minimum distance between $w$ and any $w' \in L$ of the same length.

    \begin{definition}[$\eps$-far, $\eps$-tester]\label{def:tester} A string $w$ of length $n$ is $\eps$-far from a language $L$ if the Hamming distance from $w$ to each string of length $n$ in $L$ is at least $\eps n$.
    A randomized algorithm $\mathcal{A}$ is a (two-sided) \emph{$\eps$-tester} for a language $L$, if, given query access to the input string $w$, the algorithm accepts all $w\in L$ with probability at least $\frac 2 3$ and rejects all strings $w$ that are $\eps$-far from $L$ with probability at least $\frac 23$.
    A tester is \emph{adaptive} if its queries to $w$ can depend on the answers to previous queries.
    The \emph{query complexity} of $\mathcal{A}$ is the number of queries it makes.
    A language $L$ is {\em testable} if there is a function $q:(0,1)\to\bN$ such that, for every $\eps>0$, language $L$ has an $\eps$-tester making at most $q(\eps)$ queries, independent of the input length.    
\end{definition}

We use $[n]$ to denote $\{1,2,\dots,n\}$.

\begin{definition}[Dyck languages]\label{def:dyck}
    For $m\in\mathbb{N}$, let $\Gamma_m=\{t^L,t^R : t\in\{0,\dots,m-1\}\}$, where $t^L$ and $t^R$ are the opening and closing parentheses of type $t$.
    The Dyck language $D_m\subseteq\Gamma_m^*$ is generated by $S\to\lambda\mid SS\mid t^L S\, t^R$.
\end{definition}
\newcommand{\reveal}{\ensuremath q}
\newcommand{\unrev}{\ensuremath \ell}
\newcommand{\revone}{\ensuremath s}
\newcommand{\unrevone}{\ensuremath t}
\newcommand{\eq}{\ensuremath \mathsf{EQ}}
\newcommand{\ineq}{\ensuremath \mathsf{INEQ}}
\newcommand{\hmax}{\ensuremath h_{\textsf{max}}}
\newcommand{\blk}{\ensuremath \mathsf{block}}

\section{Lower bounds for \textsc{Hidden String} and Dyck-$m$ for $m \ge 2$}
\label{sec:lb}

In this section, we prove our lower bounds for testing the Dyck-$m$ languages for $m\geq 2$ and \HSD.
We start by defining $\HS$, which only differs from \HSD in that it omits the $\diamond$ symbol.
(We need different versions for the proofs of \cref{thm:dyck-lb,thm:HSD-lb}.)
The lower bounds for \HS are stated and proved in \cref{sec:HS-lb} in \cref{thm:HS-lb-nonadaptive,thm:HS-lb-adaptive}.
Then we use these theorems to derive \cref{thm:dyck-lb,thm:HSD-lb} in \cref{sec:dyck-m-lb,sec:linear}, respectively.

\subsection{Lower bound for \HS}\label{sec:HS-lb}

In this section, we define \HS and construct a distribution framework that we then use to prove nonadaptive and adaptive lower bounds in \cref{sec:nonadaptive-lb,sec:adaptive-lb}.

\begin{definition}[\HS]\label{def:HS}
    Let $w$ be a string over an alphabet $\Sigma \cup \{ \star \}$, where $\star\notin\Sigma$.
    The {\em true-string} $T(w)$ is the string obtained by deleting all $\star$ symbols from $w$.
    Let $H(w)$ denote the string obtained by adding the $H$ superscript to all symbols of $w$.
    The {\em reverse} $\rev{w}$ is the string obtained by reversing~$w$.
    The language \HS is defined over the alphabet $\{ 0^H, 1^H, \star, 0, 1 \}$ as
    \[
        \HS := \left\{ u \circ v \mid u \in \left\{ 0^H, 1^H, \star \right\}^*, v \in \left\{ 0, 1 \right\}^*, T(u) = H(\rev{v}) \right\}.
    \]
    For a word in \HS, the prefix consisting of symbols in $\left\{ 0^H, 1^H, \star \right\}^*$ is called the {\em hidden string}, and the remaining suffix is called the {\em clear string}.
\end{definition}


Unlike in the definition of {\sf TruestringEquivalence} of \cite{FischerMS18}, our hidden string requires a superscript $H$ to prevent distance from \HS being affected by characters of $u$ counting towards those of $\rev{v}$.
This is not possible in {\sf TruestringEquivalence} as the algorithm gets separate query access to the two strings. Furthermore, to ensure \HS is a CFL, we cannot require that both strings have the same number of $\star$ symbols as in \cite{FischerMS18}. Lastly, to ensure that \HSD is a Nasu--Honda linear language, we require that the clear string contains no $\star$ symbols at all.

Our lower bound proofs are formulated via Yao's minimax principle (as stated in \cref{clm:yao}).
Namely, we describe two distributions $\Dyes$ and $\Dno$, supported over words in \HS and words $\eps$-far from \HS (with probability $1-o(1)$), respectively, and show that they are hard to distinguish by deterministic algorithms that use a small number of queries.

For convenience, in the rest of \Cref{sec:HS-lb}, we use $n$ to denote the length of the clear string rather than the total input length.
Then the full strings have length $3n$, which differs from the convention in \Cref{sec:intro} only by constant factors that are absorbed by the asymptotic notation.
Lastly, we assume $n$ is divisible by 3; the proofs extend to the general case by taking ceilings.

The pairs of distributions used in \cref{sec:adaptive-lb,sec:nonadaptive-lb} differ only in how the locations of the $\star$ symbols are chosen.
We formalize this via the notion of the {\em filter}, which is the string in $\{0,1\}^n$ of indicators for the locations of the non-star symbols in the first $n$ positions of the input string.
Next we define the pair of hard input distributions in terms of the underlying filter distributions.
The filter distributions used for our nonadaptive and adaptive lower bounds are defined in \cref{sec:nonadaptive-lb,sec:adaptive-lb}, respectively.

\begin{definition}[$\Dyes$ and $\Dno$]\label{def:dyes-dno}
    Let $v \sim \{0,1\}^n$ be sampled uniformly at random.
    Let $X_1, \dots, X_n$ be random variables supported over $\{ 0, 1 \}$ picked from an underlying {\em filter distribution}.
    Define the {\em hidden string} $u \in \left\{ 0^H, 1^H, \star \right\}^{2n}$ as follows.
    For all $i \in [n]$, if $X_i = 0$, set $u[i] = \star$ and if $X_i = 1$, set $u[2n+1-i] = \star$.
    Set the remaining characters of $u$ so that $T(u) = H(v)$.
    
    For the distribution $\Dyes$, output $u \circ \rev{v}$.
    
    Let $\Tilde{v} \in \{0,1\}^n$ be obtained from $v$ by resampling its middle third: for each $i\in(n/3,2n/3]$, set $\Tilde{v}[i]$ to be a uniformly random bit.
    For the distribution $\Dno$, output $u \circ \rev{\Tilde{v}}$.
\end{definition}

\cref{fig:filter-example} offers an illustration for these constructed strings, except there is no $\diamond$-separator.

We first show that our distributions have correct support.

\begin{lemma}\label{lem:far}
    The distributions $\Dyes$ and $\Dno$ satisfy the following.
    \begin{enumerate}
        \item $x \sim \Dyes$ is always in \HS.
        \item $y \sim \Dno$ is $\frac{1}{720}$-far from \HS~with probability $1-o(1)$.
    \end{enumerate}
\end{lemma}
\begin{proof}
    (1) follows by definition.
    We prove (2) in two steps.
    Consider a string $y \sim \Dno$ of length $3n$ of the form $u \circ \rev{\Tilde{v}}$.
    We first show that the distance of $y$ to $\HS$ is lower bounded by half the edit distance between $T(u)$ and $H(\Tilde{v})$.
    Indeed, consider the following possible modifications.
    If we change a $\star$ to $\{ 0^H, 1^H \}$ or $\{ 0, 1 \}$, this is equivalent to an insertion in $T(u)$ or $H(\Tilde{v})$, respectively.
    If we change a $\{ 0^H, 1^H \}$ or $\{ 0, 1 \}$ to $\star$, this is equivalent to a deletion in $T(u)$ or $H(\Tilde{v})$, respectively.
    If we change a $\{ 0^H, 1^H \}$ to $\{ 0, 1 \}$ or vice-versa, this is equivalent to a deletion in $T(u)$ and an insertion in $H(\Tilde{v})$, or vice-versa, respectively.
    If we change a $0^H$ to $1^H$ or vice-versa, this is equivalent to a modification in $T(u)$.
    If we change a $0$ to $1$ or vice-versa, this is equivalent to a modification in $H(\Tilde{v})$.
    Thus, every change in the input string can be reflected by at most one edit in $T(u)$ and one edit in $H(\Tilde{v})$ (total of two edits).
    Next, we show that with probability $1-o(1)$, the edit distance between $T(u)$ and $H(\Tilde{v})$ is at least $\frac{n}{120}$.
    Indeed, consider all $\binom{n/3}{n/120}^2 = o \left( 2^{16n/120} \right)$ ways of deleting $\frac{n}{120}$ characters from the middle thirds of $T(u)$ and $H(\Tilde{v})$.
    For every such possibility, the probability that the remaining strings match is $\Theta \left( 2^{-38n/120} \right)$.
    A union bound completes the argument.
    Thus the distance to $\HS$ is at least $\frac{n}{240}$ with probability $1-o(1)$, meaning the string is $1/720$-far from $\HS$.
\end{proof}

The most important part of the proof is showing that distributions $\cD^+$ and $\cD^-$ are hard to distinguish for any deterministic algorithm making a small number of queries.
Fix a deterministic tester $\cA$.
Intuitively, $\cA$ can distinguish the two distributions only if it queries a position in the middle third of the clear string and the corresponding position in the hidden string. 
We formalize this by defining matching pairs.

\begin{definition}[True index, hidden index, matching pair]\label{def:true-idx-hid-idx-match-pr}
    Given a string $w$ of length $2n$ over $\{0^H, 1^H, \star\}$ and an index $j \in [2n]$ such that $w[j] \neq \star$, the {\em true index} of $j$ in $w$ is $t(w, j) := |\{ j' \leq j \mid w[j'] \neq \star \} |$. 
    In other words, it is the index $k$ such that $w[j]$ determines $(T(w))[k]$.
    Similarly, for an index $k \in [n]$, define the {\em hidden index} $h(w, k) := j$ such that $t(w, j) = k$.

    For $j\in[2n]$ and $k\in[n]$, a pair $(j, k)$ is called {\em matching} with respect to (w.r.t.) a word $w \in \{0^H,1^H,\star\}^{2n}$ 
    if $w[j] \neq \star$ and $t(w, j) = k.$ We also define $(0,0)$ and $(2n+1,n+1)$ to be matching pairs with respect to $w$.
    The {\em index} of a {\em matching pair} $(j, k)$ is $j$.
\end{definition}

If a string $x = u \circ v$, where $u \in \{0^H,1^H,\star\}^{2n}$ and $v \in \{0,1\}^n$, is in $\HS$, then a pair $(j,k)$ with 
$j \in [2n]$ and $k \in [n]$ is matching with respect to $u$ only if $u[j] = H(v[n-k+1])$, or equivalently $x[j] = H(x[3n-k+1])$, as $H(v)$ is the reverse of $T(u).$

The {\em transcript} $\Tr{\cA, x}$ of algorithm $\cA$ on an input $x$ is the ordered sequence of pairs $(i,x[i])$, followed by a sequence of pairs $(j, X_j)$, where the value at index $i$ in the input string $x$, and the filter value $X_j$ at index $j \in [n]$ was {\em revealed} to the algorithm during its execution.
In our proofs, the value at an index or a filter value can be revealed either because it was queried or because it falls under the additional information revealed to the algorithm in order to make the proofs work.
The latter happens in order to avoid conditioning on events that are hard to analyze.
Let $\Tru{\cA}{\cD}$ be the distribution of transcripts $\Tr{\cA, x}$ when $\cA$ is executed on input $x$ drawn from distribution $\cD$.
If $E$ is an event, then $\Tru{\cA}{\cD | E}$ is the distribution of transcripts when $x$ is drawn from $\cD$, conditioned on $E$.

We define an event $\Bad$ and show that the absolute difference in acceptance probabilities of $\cA$ on inputs drawn from $\Dyes$ versus $\Dno$ is upper bounded by the probability of $\Bad$ occurring.

\begin{definition}[Bad event]
    Let $\Bad$ be the event that there exists a matching pair $(j, k)$ w.r.t. the {\em hidden string} of input $x$, such that $j \in (n/3, 2n-n/3]$ and $\Tr{\cA, x}$ has revealed the values at both indices $j$ and $3n-k+1$.
\end{definition}

\begin{lemma}\label{lem:indistinguishable-dist}
    $\left| \Pru{x \sim \Dyes}{\cA(x) \text{ accepts}} - \Pru{y \sim \Dno}{\cA(y) \text{ accepts}} \right| \le \Pr[\Bad]$.
\end{lemma}
\begin{proof}
    We first prove that the distribution on transcripts is the same under $\Dyes$ and $\Dno$ when $\Bad$ does not occur.
    \begin{claim}\label{clm:identical-views}
        $\Tru{\cA}{\Dyes \mid \overline{\Bad}} = \Tru{\cA}{\Dno \mid \overline{\Bad}}$.    
    \end{claim}
    \begin{proof}
        The locations of the $\star$ symbols and the $H$-symbols in the {\em hidden string} are distributed identically under $\Dyes$ and $\Dno$.
        The two distributions differ only in the correlations between the hidden string and the middle third of the clear string: under $\Dyes$, the corresponding clear bits agree with the bits of $v$, whereas under $\Dno$ these bits are resampled.
        
        Thus, the algorithm $\cA$ can distinguish the distributions only if there exists a matching pair $(j, k)$ w.r.t. the {\em hidden string} of input $x$, such that $k \in (n/3, 2n/3]$ and $\Tr{\cA, x}$ has revealed both indices $j$ and $3n-k+1$.
        The {\em index} $j$ of such a matching pair must lie in the interval $(n/3, 2n-n/3]$.
        Conditioned on $\overline{\Bad}$, such a matching pair does not exist in the transcript of $\cA$.
        Thus, all observed correlations lie outside the middle third and are identical under $\Dyes$ and $\Dno$.
        Therefore, conditioned on $\overline{\Bad}$, the two 
        transcript distributions are identical.
    \end{proof}
    Now the lemma follows from \cref{claim:indistinguishability-via-probability-of-bad}, whose conditions are satisfied by \cref{clm:identical-views}.
\end{proof}

\subsubsection{Nonadaptive lower bound}\label{sec:nonadaptive-lb}

In this section, we prove the following nonadaptive lower bound for testing \HS.

\begin{theorem}[\HS nonadaptive lower bound]\label{thm:HS-lb-nonadaptive}
   For all sufficiently small constant $\eps > 0$, every nonadaptive $\eps$-tester for \HS has query complexity $\Omega(n^{1/2})$ on strings of length $3n$.
\end{theorem}

We start by defining our filter distribution.

\begin{definition}[Filter distribution for nonadaptive lower bound]\label{def:filter-nonadaptive}
    The underlying filter distribution for $\Dyes$ and $\Dno$ is as follows.
    Pick a \emph{weight} $W$ uniformly at random from $\{ 0, 1, \dots, n \}$ and a uniformly random subset {\em filter} $F \subseteq [n]$ of size $W$.
    Then $X_1, \dots, X_n$ are defined as $X_i := \mathbf{1}[i \in F]$ for all $i \in [n]$.
\end{definition}

We first prove the following bound on the weight distribution of the filter.

\begin{lemma}\label{lem:polya-bound}
    If $X_1, \dots, X_n$ are as defined by the filter distribution, then for all $s \in [n]$ and all $t \in \bZ$, we have $\Pr[\sum_{i = 1}^{s} X_i = t] \le 1/s$.
\end{lemma}
\begin{proof}
    The variables $X_1, \dots, X_n$ can equivalently be generated by the following Pólya urn process:
    start with one $0$-ball and one $1$-ball, and draw bits sequentially, returning the drawn ball to the urn together with an additional ball of the same type.
    After $n$ draws, the resulting sequence has exactly the same distribution as $X$.
    Indeed, every string of weight $w\in\{0\}\cup[n]$ has probability
    \[
        \frac{w!(n-w)!}{(n+1)!}
        =
        \frac{1}{n+1}\cdot \frac{1}{\binom{n}{w}},
    \]
    which is the probability of first choosing $W \sim \{ 0, \dots, n\}$ and then sampling a uniformly random subset of $[n]$ of size $W$.
    Using the Pólya urn description, the probability that the first $s$ random variables have weight $t$ is equal to
    \[
        \binom{s}{t} \cdot \frac{t!(s-t)!}{(s+1)!} = \frac{1}{s+1} \le \frac{1}{s},
    \]
    when $t \in \{ 0 \} \cup [s]$, and $0$ otherwise.
\end{proof}

Fix a deterministic (nonadaptive) tester $\cA$ making $o(n^{1/2})$ queries.
Recall that $\Bad$ is the event that there exists a {\em matching pair} $(j, k)$ w.r.t.\ the {\em hidden string} of input $x$, such that $j \in (n/3, 2n-n/3]$ and $\cA$ has queried both indices $j$ and $3n-k+1$.

\begin{lemma}\label{lem:bad-prob}
    $\Pr[\Bad] = o(1)$.
\end{lemma}
\begin{proof}
    Consider a tuple $(j, k)$ such that $j \in (n/3, n]$ (the other side follows by symmetry) and both indices $j$ and $3n-k+1$ have been queried.
    The tuple $(j, k)$ is a {\em matching pair} only if $X_1 + \dots + X_j = k$ and $X_j = 1$.
    By \cref{lem:polya-bound}, $\Pr[X_1 + \dots + X_j = k] \le \frac{1}{j} < \frac{3}{n}$.
    So the probability that $(j, k)$ is a {\em matching pair} is at most $\frac{3}{n}$.
    The number of tuples of indices queried by $\cA$ is at most $o((n^{1/2})^2) = o(n)$, so by a union bound, the probability of $\Bad$ occurring is $o(1)$.
\end{proof}

We now complete the proof of \cref{thm:HS-lb-nonadaptive}.

\begin{proof}[Proof of \cref{thm:HS-lb-nonadaptive}]
    We use Yao's minimax principle to show that $\mathcal{D}^+$ and $\mathcal{D}^-$ are hard to distinguish with $o(n^{1/2})$ queries.
    For any deterministic (nonadaptive) tester $\mathcal{A}$ making $o(n^{1/2})$ queries,
    \begin{align*}
        \left| \Pr_{x \sim \mathcal{D}^+}[\mathcal{A}(x) \text{ accepts}]
          - \Pr_{y \sim \mathcal{D}^-}[\mathcal{A}(y) \text{ accepts}] \right|
        &\le \Pr[\mathsf{BAD}] && \text{[\cref{lem:indistinguishable-dist}]}\\
        &= o(1). && \text{[\cref{lem:bad-prob}]}
    \end{align*}
    By \cref{clm:yao}, this completes the proof of \cref{thm:HS-lb-nonadaptive}.      
\end{proof}

\subsubsection{Adaptive lower bound}\label{sec:adaptive-lb}

In this section, we prove the following adaptive lower bound for testing \HS.

\begin{theorem}[Hidden String adaptive lower bound]\label{thm:HS-lb-adaptive}
   For all sufficiently small constant $\eps > 0$, every adaptive $\eps$-tester for \HS has query complexity $\Omega(\qb)$ on strings of length $3n$.
\end{theorem}

We start by defining our filter distribution.
It divides the set $[n]$ into blocks of length $n^{3/5}$ and applies the filter distribution from \cref{def:filter-nonadaptive} independently to each block.

\begin{definition}[Filter distribution]\label{def:filter}
    Let $b := n^{3/5}$.
    Partition $[n]$ into $n/b$ consecutive \emph{blocks} $B_1, \ldots, B_{n/b}$ of length $b$, where $B_i := \{(i-1)b+1, \dots, i b\}$ for all $i \in [n/b]$.
    For each $i \in [n/b]$, pick a \emph{weight} $W_i$ uniformly at random from $\{0, 1, \dots, b\}$ and a uniformly random subset $F_i \subseteq B_i$ of size $W_i$.
    For all $j \in B_i$, set $X_j := 1$ if $j \in F_i$, and $X_j := 0$ otherwise.
    We call $(X_j)_{j \in B_i}$ the \emph{filter} of the block $B_i$.

    For $j \in [2n]$, let $\blk(j) \in [n/b]$ denote the index of the block containing $\min(j, 2n+1-j)$, and call $B_{\blk(j)}$ the \emph{block of $j$}.
\end{definition}
 
Fix a deterministic (adaptive) tester $\mathcal{A}$ making $o(\qb)$ queries.
Assume w.l.o.g.\ that when $\mathcal{A}$ queries an index $i \in [2n]$, it also queries $2n+1-i$ (this at most doubles the query
complexity).
We do this so that throughout the execution of $\mathcal{A}$, all revealed values in the \emph{hidden string} are at positions symmetric around its center.
We refer to the queries of $i$ and $2n+1-i$ together as a single query of the \emph{symmetric pair} $\{i, 2n+1-i\}$.
Generally, when an algorithm $\mathcal{A}$ makes a query, only the value at the queried index is revealed.
We strengthen $\mathcal{A}$ by revealing more information and then show that $\mathcal{D}^+$ and $\mathcal{D}^-$ are indistinguishable even for the strengthened algorithm.
Specifically, we reveal a prefix and a suffix of the \emph{hidden string} of the same length, consisting of whole blocks, together with the locations of all $\star$ symbols in every block that $\mathcal{A}$ queries, as specified next.
 
\paragraph{Additional values revealed in the hidden string.}
One of the main insights in our analysis is the strategy for revealing information to the algorithm in addition to what it obtains via queries.
We make the analysis tractable by carefully choosing the revealed information.
During the execution of algorithm $\mathcal{A}$, we maintain a value $\hmax \in \{0, b, 2b, \ldots, n\}$ that, intuitively, defines the \emph{frontiers} of the region where all values have been revealed to $\mathcal{A}$.
Since $\hmax$ is a multiple of $b$, this region consists of whole blocks.
Besides the values revealed in response to queries, we also reveal the values of all indices in the interval $[1, \hmax] \cup [2n+1-\hmax, 2n]$, that is, of all indices $j \in [2n]$ with $\blk(j) \le \hmax/b$.
We call $\hmax$ the \emph{left frontier} and $2n+1-\hmax$ the \emph{right frontier}.
 
Additionally, whenever $\mathcal{A}$ queries a symmetric pair $\{i, 2n+1-i\}$ with $i \in [\hmax+1, n]$, we reveal the filter of block $B_{\blk(i)}$, and we call this block \emph{queried}.
In other words, we reveal the locations of the $\star$ symbols at all indices $j$ with $\blk(j) = \blk(i)$, but not the values of the non-$\star$ symbols at these indices, except at the queried indices $i$ and $2n+1-i$.
 
Initially, $\hmax$ is set to $0$.
After each query, $\hmax$ may \emph{advance}, i.e., increase, 
to a multiple of $b$.
Whenever $\hmax$ advances, the values of all indices in the updated interval $[1, \hmax] \cup [2n+1-\hmax, 2n]$ are revealed to $\mathcal{A}$.
After each query made by $\mathcal{A}$, we reveal the filter of the queried block (if any), and then consider advancing $\hmax$ in two stages.
We reveal the value at the queried index only after Stage I.
 
\begin{description}
    \item[Stage I: Free advance.] Let $\gamma$ be the number of indices in the interval $[\hmax+1, n]$ queried so far.
    Arrange them in ascending order:
    $h_1 < h_2 < \ldots < h_\gamma$.
    For all $k \in [\gamma]$, let $\ell_k := \blk(h_k) - \hmax/b$ be the \emph{block distance} of $h_k$, i.e., the number of blocks from the left frontier up to and including the block of $h_k$.
    Note that $1 \le \ell_1 \le \ell_2 \le \ldots \le \ell_\gamma$.
    Let $k^\star$ be the largest index satisfying $\ell_{k^\star} \le 2k^\star$.
    If such an index $k^\star$ exists, set $\hmax$ to $b \cdot \blk(h_{k^\star})$ (equivalently, \emph{advance} $\hmax$ by $b \cdot \ell_{k^\star}$), so that $h_1, \ldots, h_{k^\star}$ now lie in $[1, \hmax]$.
    
    \item[Stage II: Direct advance.] Suppose that the newly queried index $q$ belongs to a matching pair w.r.t.\ the hidden string, and the value at $q$ has not been previously revealed but the value at the other index of the
    pair has.
    This happens in one of two ways: either $q > 2n$ and the value of the corresponding hidden string index $h(u, 3n-q+1)$ has already been revealed, or $q \in [\hmax+1, 2n-\hmax]$, the value at $q$ is not $\star$, and the corresponding clear string index $3n - t(u,q) + 1$ has already been queried.
    Let $j$ be the index of this matching pair. If $j$ is in $[\hmax+1, 2n-\hmax]$, advance $\hmax$ to $b \cdot \blk(j)$, the smallest multiple of $b$ for which $j \in [1, \hmax] \cup [2n+1-\hmax, 2n]$.
\end{description}
 
\begin{figure}[ht]
  \centering
  \begin{tikzpicture}[x=1cm, y=1cm, font=\small]
    \fill[pattern=north east lines] (0,0) rectangle (1,0.5);
    \fill[pattern=north east lines] (7,0) rectangle (8,0.5);
    \foreach \a in {1.5, 2.0, 3.0, 4.5, 5.5, 6.0}
      \fill[gray!25] (\a,0) rectangle (\a+0.5,0.5);
    \foreach \x in {0.5, 1.5, 2.0, ..., 6.5, 7.5}
      \draw[gray!85, very thin] (\x,0) -- (\x,0.5);
    \draw (0,0) rectangle (8,0.5);
    \draw[blue, thick] (1,0) -- (1,0.5);
    \draw[blue, thick] (7,0) -- (7,0.5);
    \foreach \x in {1.8, 3.2, 4.8, 5.7, 6.2}
      \draw[thick] (\x,-0.1) -- (\x,0.6);
    \draw[thick, red] (2.3,-0.1) -- (2.3,0.6);
    \draw[dotted, thick] (4,-0.3) -- (4,0.8);
    \node at (4,1.35) {Hidden String};
    \node[above] at (1.8,0.6) {$h_1$};
    \node[above] at (3.2,0.6) {$h_3$};
    \node[below, red] at (2.3,-0.1) {$h_2 = q$};
    \draw[<->] (5.0,0.75) -- (5.5,0.75) node[midway, above] {$b$};
    \draw[<->] (0,-0.35) -- (1,-0.35) node[midway, below] {$\hmax$};
    \draw[<->] (7,-0.35) -- (8,-0.35) node[midway, below] {$\hmax$};
    \draw (8.8,0) rectangle (12.8,0.5);
    \node at (10.8,1.35) {Clear String};
    \foreach \x/\lab in {9.3/c_1, 10.1/c_2, 11.6/c_3, 12.1/c_4} {
      \draw[thick] (\x,-0.1) -- (\x,0.6);
      \node[above] at (\x,0.6) {$\lab$};
    }
  \end{tikzpicture}
  \caption{An illustration of the transcript of $\mathcal{A}$. Thin gray lines
  mark block boundaries; blocks are symmetric around the center of the
  \emph{hidden string} (dotted line). All values in the hatched regions are in
  the transcript, i.e., they have been revealed to $\mathcal{A}$. The blue
  lines are the current frontiers; they lie on block boundaries. Queries
  outside of the hatched regions are denoted by vertical lines; in
  $[\hmax+1, n]$, they are indexed in ascending order. In the queried blocks
  (gray), the locations of all $\star$ symbols are revealed, but the values of
  the non-$\star$ symbols are revealed only at queried indices. The current
  query $q$ is marked in red.}
  \label{fig:transcript}
\end{figure}
 
\begin{observation}\label{obs:frontier}
    For any matching pair $(j, k)$ w.r.t.\ the hidden part of input $x$ such that the values at $j$ and $3n-k+1$ are both in the transcript of $\mathcal{A}$, the index $j$ is in $[1, \hmax] \cup [2n+1-\hmax, 2n]$.
\end{observation}
 
Next, we define an event $\mathsf{BAD}_1$, which is necessary for
$\mathsf{BAD}$ to occur, and show in \cref{lem:bad1-prob} that $\mathsf{BAD}_1$
occurs with small probability.
 
\begin{definition}[Bad event]\label{def:bad1}
    Let $\mathsf{BAD}_1$ be the event that $\hmax$ is greater than $n/3$ at any point during the execution of $\mathcal{A}$.
\end{definition}
 
\begin{lemma}\label{lem:bad-le-bad1}
    $\Pr[\mathsf{BAD}] \le \Pr[\mathsf{BAD}_1]$.
\end{lemma}
 
\begin{proof}
    Recall that $\mathsf{BAD}$ is the event that there exists a matching pair $(j,k)$ w.r.t.\ the hidden part of input $x$, such that $j \in (n/3, 2n-n/3]$ and $\mathcal{T}[\mathcal{A}, x]$ has revealed the values at both indices $j$ and $3n-k+1$. 
    Conditioned on $\overline{\mathsf{BAD}_1}$ occurring, $\hmax$ is at most $n/3$, which by \cref{obs:frontier} then implies that the transcript of $\mathcal{A}$ does not have such a matching pair.
    Consequently, $\mathsf{BAD}_1$ must occur for $\mathsf{BAD}$ to occur, and thus $\Pr[\mathsf{BAD}] \le \Pr[\mathsf{BAD}_1]$.
\end{proof}
 
We use the following bound on the distribution of sums of independent uniform random variables.
We prove it in \cref{sec:bounds}.
 
\begin{restatable}{lemma}{antiConcentration}\label{lem:block-anticoncentration}
    There exists an absolute constant $C_0 \ge 1$ such that the following holds for all $s \in \mathbb{N}$ and large enough $b$.
    If $W_1, \ldots, W_s$ are independent random variables chosen uniformly from $\{0, 1, \ldots, b\}$, then for all
    $t \in \mathbb{Z}$, we have $\Pr \left[ \sum_{i=1}^{s} W_i = t \right] \le \frac{C_0}{b \sqrt{s}}$.
\end{restatable}

We use this inequality to bound the probability of $\mathsf{BAD}_1$ below.

\begin{lemma}\label{lem:bad1-prob}
  $\Pr[\mathsf{BAD}_1] = o(1)$.
\end{lemma}
 
\begin{proof}
    For $\hmax$ to be greater than $n/3$, its total \emph{advance} must be greater than $n/3$.
    We divide the \emph{advances} into \emph{free} and \emph{direct} advances, and then bound them separately.

    \begin{claim}\label{clm:free-adv}
        The total {\em free advance} in an execution of $\cA$ is always $o(n)$.
    \end{claim}
    \begin{proof}
        Observe that any time there is a \emph{free advance}, $\hmax$ advances by $b \cdot \ell_{k^\star} \le 2k^\star b$, and the $k^\star$ queried indices $h_1, \ldots, h_{k^\star}$ move into $[1, \hmax]$, where they stay.
        So every queried symmetric pair contributes to \emph{free advance}
        at most once, by at most $2b$.
        In total, the \emph{free advance} is thus at most $o(\qb) \cdot 2b = o(\qb) \cdot 2n^{3/5} = o(n)$.    
    \end{proof}
    
    Next we bound the total \emph{direct advance}.
    To do so, we first compute the expected \emph{direct advance} for an arbitrary query $q$, conditioned on the transcript.
    Consider the moment right after Stage~I for query $q$, before the value at $q$ is revealed. 
    Everything up to the \emph{frontiers} is revealed, along with the filters of the queried blocks and the values at other queried indices.
    Let $h_1 < h_2 < \cdots < h_\gamma$ and $\ell_1 \le \cdots \le \ell_\gamma$ be defined as in Stage~I, for the current value of $\hmax$ (right after \emph{free advance}). 
    Let $c_1, \ldots, c_\delta$ be the indices of the \emph{clear string} queried before $q$; note that $\gamma, \delta = o(\qb)$. 
    Since this is after the \emph{free advance} stage,\footnote{If $\hmax$ advanced during the \emph{free advance} stage, the remaining queries between the frontiers were re-indexed and their block distances decreased by $\ell_{k^\star}$.
   However, $\ell_k \le 2k$ is still \emph{not satisfied} under the new indexing.
    Otherwise, the index $h_{k^\star + k}$ would have satisfied $\ell_{k^\star+k} \le 2(k^\star + k)$ under the original indexing, despite $k^\star + k$ being greater than $k^\star$.}
    we have
    \begin{equation}\label{eq:lk-min-length}
        \ell_k > 2k \quad \text{for all } k \in [\gamma].
    \end{equation}
    We call a block \emph{revealed} if it is one of the first $\hmax/b$ blocks or it is queried, and \emph{unrevealed} otherwise.
    Let $R$ be the event that the revealed filter bits take their observed values.
    
    Consider $k \in [\gamma]$.
    By construction, exactly one of the indices $h_k$ and $2n+1-h_k$ holds a non-$\star$ symbol, which is determined by $X_{h_k}$, whose value is revealed in the transcript.
    We say that the symmetric pair $\{h_k, 2n+1-h_k\}$ is \emph{matched} with an index $c$ of the \emph{clear string} if $(j, 3n+1-c)$ is a matching pair, where $j \in \{h_k, 2n+1-h_k\}$ is the non-$\star$ index of the symmetric pair.
    
    \begin{claim}\label{clm:single-match-prob}
        Conditioning only on the revealed filter bits in the hidden string,
        \[
            \Pr \left[\{h_k, 2n+1-h_k\} \text{ is matched with } c \mid R \right]
            \le \frac{\sqrt{2}\, C_0}{b\sqrt{\ell_k}}.
        \]
    \end{claim}
    \begin{proof}
        The true index of the non-$\star$ symbol of the symmetric pair $\{h_k, 2n+1-h_k\}$ is
        \begin{align*}
            t(u, h_k) &= \textstyle\sum_{i=1}^{h_k} X_i
            && \text{when } X_{h_k} = 1, \text{ and}\\
            t(u, 2n+1-h_k) &= n + 1 - \textstyle\sum_{i=1}^{h_k} (1 - X_i)
            && \text{when } X_{h_k} = 0.
        \end{align*}
        Each of these is the sum, over the blocks preceding $B_{\blk(h_k)}$, of $W_i$ (respectively, $b - W_i$), plus a contribution of the block $B_{\blk(h_k)}$ itself.
        Except for the unrevealed blocks, all the other terms are already determined by the transcript.
        Of the $\ell_k - 1$ blocks strictly between the left frontier and $B_{\blk(h_k)}$, only the blocks of $h_1, \ldots, h_{k-1}$ have already been queried.
        So, by \eqref{eq:lk-min-length}, at least $\ell_k - k > \ell_k/2$ of them are unrevealed.
        Since $W_i$ is uniform over $\{0, \ldots, b\}$, so is $b-W_i$, and thus the true index of the non-$\star$ symbol equals a number determined by the transcript, plus or minus a sum of more than $\ell_k/2$ independent random variables chosen uniformly from $\{0, \ldots, b\}$.
        By \cref{lem:block-anticoncentration}, for every index $c$ of the \emph{clear string},
        \[
            \Pr \left[\{h_k, 2n+1-h_k\} \text{ is matched with } c \mid R \right]
            \le \frac{C_0}{b\sqrt{\ell_k/2}} = \frac{\sqrt{2}\, C_0}{b\sqrt{\ell_k}}. \qedhere
        \]   
    \end{proof}
    
    Additionally, we know that no queried index of the \emph{hidden string} that lies between the \emph{frontiers} forms a matching pair with a queried index
    of the \emph{clear string}.
    Such a constraint affects the filter distribution of unrevealed blocks, so we bound their conditional probability below.

    \begin{claim}\label{clm:conditioned-match=prob}
        For any event $E$ determined by the filters of the unrevealed blocks,
        \[
            \Pr[E \mid \text{transcript}] \le 2 \Pr [E \mid R].
        \]
    \end{claim}
    \begin{proof}
        Formally, let $I$ be the event that, for all $k \in [\gamma]$ with $q \notin \{h_k, 2n+1-h_k\}$, the pair $\{h_k, 2n+1-h_k\}$ is not matched with any of $c_1, \ldots, c_\delta$.
        By a union bound over the conditions defining $I$, together with \eqref{eq:lk-min-length} and \cref{clm:single-match-prob},
        \[
            \Pr \left[\,\overline{I} \mid R\,\right]
            \le \delta \sum_{k=1}^{\gamma}
            \frac{\sqrt{2}\, C_0}{b\sqrt{\ell_k}}
            < \delta \sum_{k=1}^{\gamma} \frac{C_0}{b\sqrt{k}}
            \le \delta \cdot \frac{C_0}{b} \sum_{k = 1}^{\gamma} \frac{1}{\sqrt{k}}
            \le \delta \cdot \frac{C_0}{b} \cdot 2 \sqrt{\gamma}
            = \frac{o(\qb \cdot n^{1/5})}{n^{3/5}} = o(1).
        \]
        Hence, $\Pr[I \mid R] \ge 1/2$ for all sufficiently large $n$.
        For any event $E$ determined by the filters of the unrevealed blocks, we get
        \[
            \Pr[E \mid \text{transcript}] = \Pr[E \mid I, R]
            = \frac{\Pr[E \cap I \mid R]}{\Pr[I \mid R]}
            \le \frac{\Pr[E \mid R]}{\Pr[I \mid R]} \le 2 \Pr[E \mid R]. \qedhere
        \]
    \end{proof}
    
    Finally we bound the expected \emph{direct advance} for the query $q$, conditioned on the transcript.

    \begin{claim}
        The expected \emph{direct advance} for a query $q$ is $o(n^{3/5})$.
    \end{claim}
    \begin{proof}
        We split the analysis into two cases.
        
        \textbf{Case 1: $q \le 2n$.}
        Assume that the symmetric pair containing $q$ lies between the \emph{frontiers} after Stage~I (otherwise, its
        \emph{direct advance} is $0$).
        Then this pair is $\{h_r, 2n+1-h_r\}$ for some $r \in [\gamma]$.
        A \emph{direct advance} occurs only if this pair is matched with one of $c_1, \ldots, c_\delta$, and then $\hmax$ advances by $b \cdot \ell_r$.
        By \cref{clm:single-match-prob,clm:conditioned-match=prob}, and a union bound over the $\delta = o(\qb)$ queries in
        the \emph{clear string}, this happens with probability at most
        $o(\qb) \cdot \frac{2\sqrt{2}\, C_0}{b\sqrt{\ell_r}}$.
        Since $\ell_r \le n/b$, the expected \emph{direct advance} is at most
        \[
            b\,\ell_r \cdot o(\qb) \cdot \frac{2\sqrt{2}\, C_0}{b\sqrt{\ell_r}}
            = o(\qb) \cdot 2\sqrt{2}\, C_0\, \sqrt{\ell_r}
            \le o(\qb) \cdot 2\sqrt{2}\, C_0\, \sqrt{n/b}
            = o(\qb) \cdot n^{1/5} = o(n^{3/5}).
        \]
        
        \textbf{Case 2: $q > 2n$.} A \emph{direct advance} occurs only if some pair $\{h_k, 2n+1-h_k\}$ with $k \in [\gamma]$ is matched with $q$, and then
        $\hmax$ advances by $b \cdot \ell_k$.
        By \cref{clm:single-match-prob} and \cref{clm:conditioned-match=prob}, the expected \emph{direct advance} is at most
        \[
            \sum_{k=1}^{\gamma} b\,\ell_k \cdot
            \frac{2\sqrt{2}\, C_0}{b\sqrt{\ell_k}}
            = 2\sqrt{2}\, C_0 \sum_{k=1}^{\gamma} \sqrt{\ell_k}
            \le o(\qb) \cdot 2\sqrt{2}\, C_0\, \sqrt{n/b} = o(n^{3/5}). \qedhere
        \]
    \end{proof}
    
    As this bound holds for an arbitrary transcript, it also bounds the unconditional expectation.
    By linearity of expectation over the at most $o(\qb)$ queries, the expected total \emph{direct advance} is at most
    $o(\qb) \cdot o(n^{3/5}) = o(n)$.
    Thus, the expected total \emph{advance} is $o(n)$. By Markov's inequality, the probability that the total advance exceeds $n/3$ is $o(1)$.
\end{proof}
 
We now complete the proof of \cref{thm:HS-lb-adaptive}.
 
\begin{proof}[Proof of \cref{thm:HS-lb-adaptive}]
  We use Yao's minimax principle to show that $\mathcal{D}^+$ and
  $\mathcal{D}^-$ are hard to distinguish with $o(n^{2/5})$ queries. For any
  deterministic (adaptive) tester $\mathcal{A}$ making $o(n^{2/5})$ queries,
  \begin{align*}
    \left| \Pr_{x \sim \mathcal{D}^+}[\mathcal{A}(x) \text{ accepts}]
      - \Pr_{y \sim \mathcal{D}^-}[\mathcal{A}(y) \text{ accepts}] \right|
    &\le \Pr[\mathsf{BAD}] && \text{[\cref{lem:indistinguishable-dist}]}\\
    &\le \Pr[\mathsf{BAD}_1] && \text{[\cref{lem:bad-le-bad1}]}\\
    &= o(1). && \text{[\cref{lem:bad1-prob}]}
  \end{align*}
  By \cref{clm:yao}, this completes the proof of \cref{thm:HS-lb-adaptive}.
\end{proof}

\subsection{Implications for Dyck languages}\label{sec:dyck-m-lb}

We now prove \Cref{thm:dyck-lb-nonadaptive,thm:dyck-lb}.
The proofs of \Cref{thm:HS-lb-nonadaptive,thm:HS-lb-adaptive} show that every deterministic nonadaptive algorithm making $o(n^{1/2})$ queries, and every deterministic adaptive algorithm making $o(n^{2/5})$ queries, distinguishes $\mathcal{D}^+$ from $\mathcal{D}^-$ with advantage $o(1)$, where $\mathcal{D}^{\pm}$ use the filter of \Cref{def:filter-nonadaptive} and of \Cref{def:filter}, respectively.
We transfer these bounds to $D_2$ using the encoding of \cite[Theorem~3]{ParnasRR03} and \cite[Lemma~2.5]{FischerMS18}.
Let $\rho$ be the block map
\[
    \begin{array}{lll}
        \rho(0^H)=0^L0^L, & \rho(1^H)=1^L1^L, & \text{(hidden bits: opening)}\\
        \rho(\star)=0^L0^R, & & \text{(star: matched pair)}\\
        \rho(0)=0^R0^R, & \rho(1)=1^R1^R & \text{(clear bits: closing)}
    \end{array}
\]
extended to strings symbol by symbol, and let $\rho(\mathcal{D})$ denote the distribution of $\rho(z)$ for $z\sim\mathcal{D}$. 
We do not use $\rho$ as a reduction for all inputs, since it does not preserve distances to $\textsf{HiddenString}$: a word of $\textsf{HiddenString}$ cannot contain stars in its clear part, unlike the instances of Truestring Equivalence in \cite{FischerMS18}. Instead, we use the following property of the hard distributions.
 
\begin{lemma}\label{lem:encoding}
    The distributions $\Dyes$ and $\Dno$ satisfy the following.
    \begin{enumerate}
        \item $\rho(x) \sim \rho(\Dyes)$ is supported on $D_2$.
        \item $\rho(y) \sim \rho(\Dno)$ is $\frac{1}{1440}$-far from $D_2$ with probability $1-o(1)$.
    \end{enumerate}
\end{lemma}
\begin{proof}
    (1) Let $x=u\circ c$ be in the support of $\mathcal{D}^+$, so that $u\in\{0^H,1^H,\star\}^{2n}$ has exactly $n$ non-$\star$ symbols and $c\in\{0,1\}^n$.
    In $\rho(x)$, every $\star$ becomes an adjacent matched pair $0^L0^R$.
    Deleting these pairs leaves the opening parentheses $t^L t^L$ for the symbols $t^H$ of $T(u)$, in order, followed by the closing parentheses $t^R t^R$ for the symbols $t$ of $c$, in order.
    This string is in $D_2$ if $T(u)=H(c^{\mathrm{rev}})$, which holds when $x\sim\mathcal{D}^+$.
     
    (2) Now let $y\sim\mathcal{D}^-$, so that $T(u)=H(v)$ and $c=\tilde v^{\mathrm{rev}}$.
    As shown in the proof of \Cref{lem:far}, with probability $1-o(1)$, the edit distance between $T(u)$ and $H(\tilde v)$ is at least $n/120$.
    Let $\mathsf{LCS}$ be the length of the longest common subsequence between $T(u)$ and $H(\tilde{v})$.
    One can turn $T(u)$ into $H(\tilde v)$ by deleting the $n-\mathsf{LCS}$ symbols outside a longest common subsequence and inserting $n-\mathsf{LCS}$ symbols. 
    Hence $\mathsf{LCS}(T(u),H(\tilde v))\le n-n/240$.
    It remains to show that $\mathrm{dist}(\rho(y),D_2)\ge n-\mathsf{LCS}(T(u),H(\tilde v))$, since then $\rho(y)$ is $\frac{n/240}{6n}=\frac{1}{1440}$-far from $D_2$.
     
    Let $w\in D_2$ have length $6n$, and let $B$ be the set of positions where $w$ and $\rho(y)$ differ.
    Pair the positions of $w$ so that each opening parenthesis is paired with a later closing parenthesis of the same type and no two pairs $(p,q)$ and $(p',q')$ cross, that is, satisfy $p<p'<q<q'$.
    Let $p(q)$ denote the partner of a closing position $q$. 
    The last $2n$ positions of $\rho(y)$, which encode $c$, are closing parentheses.
    Call such a position $q$ \emph{good} if $q\notin B$, $p(q)\notin B$, and $p(q)$ lies in the encoding of a non-$\star$ symbol of $u$.
    If $q\notin B$ is not good, then either $p(q)\in B$, or $p(q)$ is the first position of an encoding $0^L0^R$ of a star and $p(q)+1\in B$.
    The first case includes $p(q)>4n$, because $w[p(q)]$ is an opening parenthesis while $\rho(x)[p(q)]$ is a closing one.
    In the second case, if $p(q)+1\notin B$, then $w[p(q)+1]=0^R$ would be paired with some $p'<p(q)$, and the pairs $(p',p(q)+1)$ and $(p(q),q)$ would cross, which is a contradiction.
    We charge a position $q$ that is not good to $q$ if $q\in B$, to $p(q)$ if $q\notin B$ and $p(q)\in B$, and to $p(q)+1$ otherwise.
    Then every position of $B$ is charged at most twice, so at least $2n-2|B|$ of the last $2n$ positions are good.
     
    For a good position $q$, the position $p(q)$ encodes a symbol $t^H$ of $u$, and $q$ encodes the symbol $t$ of $c$, since paired parentheses have the same type.
    All left ends $p(q)$ precede all right ends $q$, so non-crossing pairs are nested.
    Hence, sorted by $p(q)$, the good pairs encode nondecreasing positions in $u$ and nonincreasing positions in $c$.
    Each position repeats at most twice, and only in consecutive pairs, since every symbol is encoded by two parentheses. 
    Keeping every other good pair therefore gives at least $n-|B|$ pairs whose positions strictly increase in $u$ and strictly decrease in $c$.
    They form a common subsequence of $T(u)$ and $H(c^{\mathrm{rev}})=H(\tilde v)$ of length at least $n-|B|$.
\end{proof}

\begin{proof}[Proofs of \Cref{thm:dyck-lb,thm:dyck-lb-nonadaptive}]
    Fix $m\ge2$ and $\varepsilon\le 1/1500$.
    A string over the alphabet of $D_2$ has the same distance to $D_m$ as to $D_2$, because relabeling every parenthesis of type $t\ge2$ as type~$0$ maps $D_m$ into $D_2$ and does not change any position where a word agrees with the string.
    So \Cref{lem:encoding} also holds for $D_m$.
    Let $\mathcal{B}$ be a deterministic algorithm that makes $q$ queries on inputs of length $N=6n$.
    Running $\mathcal{B}$ on $\rho(z)$ and answering each query to position $i$ by querying $z[\lceil i/2\rceil]$ gives an algorithm on inputs $z$ of length $3n$.
    This algorithm makes $q$ queries, is nonadaptive if $\mathcal{B}$ is, and distinguishes $\mathcal{D}^+$ from $\mathcal{D}^-$ with the same advantage with which $\mathcal{B}$ distinguishes $\rho(\mathcal{D}^+)$ from $\rho(\mathcal{D}^-)$.
    By the bounds stated at the beginning of this section, this advantage is $o(1)$ if (a) $\mathcal{B}$ is nonadaptive and $q=o(n^{1/2})$, or (b) $\mathcal{B}$ is adaptive and $q=o(n^{2/5})$.
    \Cref{lem:encoding} and \Cref{clm:yao} now prove both theorems for inputs of length $N=6n$.
\end{proof}

\subsection{Implications for linear languages}
\label{sec:linear}

In this section, we prove \Cref{thm:HSD-lb}.

\begin{proof}[Proof of \Cref{thm:HSD-lb}]
    The proofs of \Cref{thm:HS-lb-adaptive,thm:HS-lb-nonadaptive}, which show that \HS requires $\Omega(n^{1/2})$ queries to test nonadaptively and $\Omega(\qb)$ queries to test adaptively, extend to \HSD after inserting a $\diamond$ between the hidden and clear strings in the distributions. 
    
    It remains to show that \HSD is a Nasu--Honda deterministic linear
    language. A linear language is Nasu--Honda deterministic if it is generated by a grammar $G=(V,\Sigma,R,S)$ whose productions satisfy the
    following conditions (where $V$ is the set of non-terminals, $\Sigma$ is the set of the terminals, 
    $R$ is the set of production rules, and $S$ is the start symbol): 
    \begin{enumerate}
        \item Every production is of the form $A\to aBu$ or $A\to a$, where
        $a\in\Sigma$, $B\in V$, and $u\in\Sigma^*$.
        \item For all $A\in V$, $a\in\Sigma$, and
        $\alpha,\beta\in V\Sigma^*\cup\{\lambda\}$, if 
        $A\to a\alpha$ and $A\to a\beta$ are in $R$, then $\alpha=\beta$.
    \end{enumerate}
    
    The language \HSD is generated by the grammar
    \[
        S \to \star S \mid 0^H S 0 \mid 1^H S 1 \mid \diamond .
    \]
    Each production has the required form, and each terminal starts at most one rule.
    So the determinism condition is satisfied, and \HSD is a Nasu--Honda deterministic linear language.
\end{proof}
\section{Testing excursion languages and the Dyck-1 language}\label{sec:excursion-languages}

In this section, we give our algorithms for the excursion languages and Dyck-1.
We start by formally defining excursion languages in \cref{def:excursion-language} and stating the guarantees of our tester for them in \cref{thm:excursion-up}. \cref{sec:repair} states and proves the main ingredient in our analysis, the Repair Lemma. 
\cref{sec:analysis-excursion-tester} uses the Repair Lemma to analyze the tester.
Finally, \cref{sec:D1-tester} explains how to use this tester for Dyck-1.

\begin{definition}[Excursion languages]\label{def:excursion-language}
    Let $\ell,r\in\bN$ and $\Sigma_{\ell,r} := \{-\ell,-\ell+1,\ldots,-1,0,1,\ldots,r\}$.
    For a word $w\in\Sigma_{\ell,r}^n$, define the partial sums $s_0 (w) = 0$ and $s_j(w)=\sum_{i\in[j]} w[i]$ for all $j\in[n]$, the final sum $s(w)=s_n(w)$, and the minimum sum $m(w)=\min_{0\le j\le n}s_j(w)$.
    Define the {\em excursion language} $L_{\ell,r}$ as
    \[
        L_{\ell,r}
        =
        \{w\in \Sigma_{\ell,r}^* : s(w)=0 \text{ and } m(w)=0\}.
    \]
\end{definition}

\begin{theorem}\label{thm:excursion-up}
For every fixed $\ell,r\in\bN$ and for all $\eps \in (0,1)$, the language $L_{\ell,r}$ has a nonadaptive
$\eps$-tester that makes $O(\rho^2/\eps^2)$ queries, where $\rho=\max\{\ell,r\}$.
\end{theorem}

For small $n$, specifically when $n\le C\rho^2/\eps^2$, where $C$ is sufficiently large, the tester queries
the entire input and decides membership exactly. In the pseudocode and
analysis below, we assume $n>C\rho^2/\eps^2$.
Let $\Delta(w)=s(w)-2m(w)$. \Cref{alg:excursion-tester} samples a subsequence $\tilde w$ from input $w$  and estimates $\Delta(w)$ by computing $\Delta(\tilde w)$ and rescaling. It then rejects if this estimate is too high.

\begin{algorithm}[H]
    \caption{Tester for excursion languages}\label{alg:excursion-tester}
    \begin{algorithmic}[1]
        \State \textbf{input: } Query access to a word $w\in\Sigma_{\ell,r}^n$ and an error parameter $\eps>0$
        \State $\rho\gets \max\{\ell,r\}$
        \State $\alpha\gets \frac{C\rho^2}{\eps^2}$ for a sufficiently large universal constant $C$
        \State Sample a set $Q\subseteq[n]$ by including each index independently with probability $p=\frac \alpha n$
        \State Query $w[i]$ for every $i\in Q$
        \State Let $\tilde w$ be the word obtained by setting $\tilde w[i]=w[i]$ for $i\in Q$ and $\tilde w[i]=0$ for $i\notin Q$
        \State $\widetilde{\Delta}\gets s(\tilde w)-2m(\tilde w)$
        \State \textbf{accept} if $\frac{n}{\alpha}\widetilde{\Delta}<\frac{\eps n}{2}$ else \textbf{reject}
    \end{algorithmic}
\end{algorithm}

As written, \Cref{alg:excursion-tester} uses $O(\alpha)$ queries in
expectation. To obtain a worst-case query bound, one can impose a query cap:
stop and reject if more than, say, $10\alpha$ indices are sampled. By a
Chernoff bound, the query cap is triggered with exponentially small probability
in $\alpha$, so it changes the success probability only negligibly. We omit
this query cap from the pseudocode.

\subsection{The Repair Lemma}\label{sec:repair}
In this section, we prove the following structural lemma used in the analysis of \cref{alg:excursion-tester}.
\begin{figure}
    \centering
    \begin{subfigure}[b]{0.49\textwidth}
        \centering
        \includegraphics[width = \textwidth]{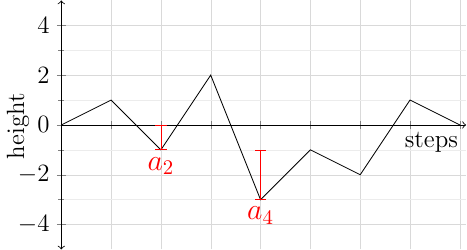}
        \caption{Stage 1}
        \label{fig:stage-1}
    \end{subfigure}
    \hfill
    \begin{subfigure}[b]{0.49\textwidth}
        \centering
        \includegraphics[width = \textwidth]{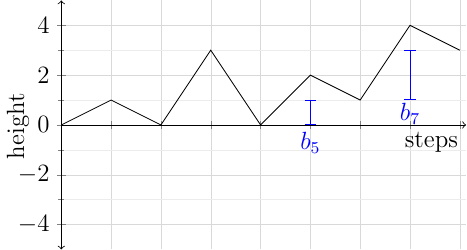}
        \caption{Stage 2}
        \label{fig:stage-2}
    \end{subfigure}
    \caption{An example of repairing the word
    $(1,-2,3,-5,2,-1,3,-1)$. 
    In stage 1, we increase $w[2]$ by $1$ and $w[4]$ by $2$ to get $u =(1,-1,3,-3,2,-1,3,-1)$.  
    In stage 2, we reduce $u[5]$ and $u[7]$ by $1$ and $2$ respectively to get the repaired word $v = (1, -1, 3, -3, 1, -1, 1, -1)$.
    }
    \label{fig:repair-lemma}
\end{figure}

\begin{lemma}[Repair lemma]\label{lem:repair}
For all $w\in \Sigma_{\ell,r}^n$,
$\dist(w,L_{\ell,r}) \le s(w)-2m(w)$.
\end{lemma}

\begin{proof}
We can view the string $w$ as a walk on the integer line. For $i\in[n]$,
the symbol $w[i]$ represents step $i$ of the walk, and the partial sum $s_i(w)$
is the height after $i$ steps. Membership of string $w$ in the language
$L_{\ell,r}$ means that the walk never goes below $0$ and ends at $0$.
There are two ways in which $w$ can fail to be in the language: (1) the walk
goes below $0$; (2) it ends above or below $0$. We fix these problems in two separate stages, illustrated in \Cref{fig:repair-lemma}.

Since $s_0(w)=0$, we have $m(w)\le 0$. Let $M=-m(w)$. We use the convention that
$[0]=\emptyset$.

In the {\em first stage,} we repair all negative prefixes. For each height
$h\in [M]$, let
\[
    i_h=\min\{j\in[n]: s_j(w)\le -h\}
\]
be the first step on which the walk reaches height $-h$ or lower. As $w[i]$ can be less than $-1$ (if $\ell > 1$), multiple heights can be reached for the first time on the same step. For each
index $i\in[n]$, we define the amount we will add to symbol $w[i]$ as the number of negative heights reached for the first time on step $i$, i.e.,
\[
    a_i=\bigl|\{h\in[M]: i_h=i\}\bigr|.
\]
We obtain a new word $u$ from $w$ by replacing each symbol $w[i]$ with
$u[i]=w[i]+a_i$.

Only positions with $a_i>0$ are changed, so this stage changes at most $M$
symbols. We now check that $u$ is still over the alphabet $\Sigma_{\ell,r}$.
If $a_i>0$, then step $i$ is the first to reach some negative heights, and these heights are between $s_i(w)$ and $s_{i-1}(w).$
Therefore, $a_i\le s_{i-1}(w)-s_i(w)=-w[i]$ and, consequently,
$u[i]=w[i]+a_i\le 0$. In addition, $u[i]\ge w[i]\ge -\ell$, showing that
$u[i]\in\Sigma_{\ell,r}$.

Next, we argue that the walk determined by $u$ never goes below $0$.
Observe that
\[
    s_j(u)
    =
    s_j(w)+\bigl|\{h\in[M]: i_h\le j\}\bigr|.
\]
If $s_j(w)\ge 0$, then certainly $s_j(u)\ge 0$. If $s_j(w)<0$, then for every
$h\in[-s_j(w)]$, the walk has reached height $-h$ by time $j$, so $i_h\le j$.
Hence, $s_j(u)\ge s_j(w)+(-s_j(w))=0$. Thus, all partial sums of $u$ are
nonnegative. Moreover, the final sum of $u$ is $s(u)=s(w)+M=s(w)-m(w)$.

Let $H=s(u) =s(w)-m(w)$. Since $m(u)\leq 0$ and all partial sums of $u$ are nonnegative, $H\ge 0$.
The {\em second stage} reduces the final sum from $H$ to $0$ without
creating a negative prefix. For every height $h\in[H]$, let
\[
    j_h
    =
    \max\{j\in[n]: s_{j-1}(u)<h\le s_j(u)\}.
\]
This is the last step on which the walk defined by $u$ crosses height $h$
upward. For each index $j\in[n]$, define the amount we will subtract from
symbol $u[j]$ as
\[
    b_j=\bigl|\{h\in[H]: j_h=j\}\bigr|.
\]
We obtain a word $v$ from $u$ by replacing each symbol $u[j]$ with
$v[j]=u[j]-b_j$.

The second stage changes at most $H$ symbols.
Now we argue that $v$ is over
the alphabet $\Sigma_{\ell,r}$. Indeed, if $b_j>0$, then the heights assigned
to $j$ are heights crossed upward during step $j$, so
$b_j\le s_j(u)-s_{j-1}(u)=u[j]$. Therefore,
$0\le v[j]=u[j]-b_j\le u[j]\le r$.

Next we show $s(v) = 0$.
By definition of $v$,
\[
    s_j(v)
    =
    s_j(u)-\bigl|\{h\in[H]: j_h\le j\}\bigr|.
\]
We claim that $s_j(v)\ge 0$ for all $j\in [n]$. 
Fix $j$. For all $h\in[H]$ such that $j_h\le j$, since
$j_h$ is the last upward crossing of height $h$, and
since the final height is $s(u)=H\ge h$, the walk determined by $u$ cannot be
below height $h$ at time $j$; otherwise, it would have to cross height $h$
upward again after time $j$. Hence $s_j(u)\ge h$.
Therefore,
\[
    \{h\in[H]: j_h\le j\}\subseteq [s_j(u)],
\]
which implies $\bigl|\{h\in[H]: j_h\le j\}\bigr| \le s_j(u)$.
So,
$$s_j(v)=s_j(u)-\bigl|\{h\in[H]: j_h\le j\}\bigr| \geq s_j(u) - 
s_j(u)\ge 0.$$
Finally,
$s(v)=s(u)-H=0$. Thus, $v\in L_{\ell,r}$.

The total number of symbol changes is at most
$M+H=(-m(w))+(s(w)-m(w))=s(w)-2m(w)$. Therefore,
$\dist(w,L_{\ell,r})\le s(w)-2m(w)$, as claimed.
\end{proof}

\subsection{Analysis of the tester for excursion languages}\label{sec:analysis-excursion-tester}

We now prove correctness of \Cref{alg:excursion-tester}. Soundness follows
from the repair lemma: if a word is far from $L_{\ell,r}$, then
$\Delta(w)=s(w)-2m(w)$ is large, and the sampled estimate is unlikely to
substantially underestimate it. Completeness follows by decomposing the sampled
walk into its expected value plus independent mean-zero sampling errors. The
expected sampled walk is just a scaled copy of the original walk, which stays
nonnegative and ends at $0$ when $w\in L_{\ell,r}$; the error terms are
controlled by Kolmogorov's inequality and Chebyshev's inequality.

\begin{lemma}[Soundness]\label{lem:excursion-soundness}
If $w$ is $\eps$-far from $L_{\ell,r}$, then
$\Pr[\text{\Cref{alg:excursion-tester} rejects } w]\ge 2/3$.
\end{lemma}

\begin{proof}
Assume that $w$ is $\eps$-far from $L_{\ell,r}$. By
\Cref{lem:repair}, for every $w\in\Sigma_{\ell,r}^n$,
\[
    \dist(w,L_{\ell,r})\le \Delta(w).
\]
Therefore, $\Delta(w)\ge \eps n$. We will show that the sampled quantity $\widetilde{\Delta}$ is unlikely to substantially underestimate $\Delta(w)$. This is enough for soundness, because the tester rejects exactly when the scaled sampled quantity is large.

Fix $j\in [n]$. Define
$\widehat s_j=\frac{n}{\alpha}\sum_{i\in[j]}\tilde w[i]$ and let
$\widehat s_0=0$. Let $X_i$ be the indicator of the event that index $i$ is
sampled. Then $\Pr[X_i=1]=p=\frac{\alpha}{n}$, and
\[
    \widehat s_j
    =
    \frac{n}{\alpha}\sum_{i\in[j]}w[i]X_i .
\]
Then $\mathbb E[\widehat s_j]=s_j(w)$.
We will bound the probability that $\widehat s_j$ deviates from $s_j(w)$. Let $Y_i=\frac{n}{\alpha}w[i]X_i$ for $i\in[j]$.
Since $|w[i]|\le \rho$, 
\[
    \sum_{i\in[j]}\Var(Y_i)
    \le
    \frac{n^2}{\alpha^2}\cdot p\sum_{i\in[j]}w[i]^2
    \le
    \frac{\rho^2 n^2}{\alpha}.
\]
Also, $|Y_i-\mathbb E[Y_i]|\le \frac{\rho n}{\alpha}$. Therefore, by
Bernstein's inequality (\Cref{thm:bernstein-ineq}),
\[
    \Pr\left[\left|\widehat s_j-s_j(w)\right|\ge \frac{\eps n}{6}\right]
    \le
    2\exp\left(-\Omega\left(\frac{\alpha\eps^2}{\rho^2}\right)\right).
\]
Since $\alpha=\frac{C\rho^2}{\eps^2}$, this probability is at most
$\frac{1}{24}$ for a sufficiently large constant $C$. For $j=0$, the same
bound is trivial because $\widehat s_0=s_0(w)=0$.

Let $j^*\in\{0,\ldots,n\}$ be such that $s_{j^*}(w)=m(w)$. Applying the
preceding bound to $j=n$ and to $j=j^*$, and taking a union bound, we get
that with probability at least $\frac{11}{12}$, both
$\widehat s_n\ge s(w)-\frac{\eps n}{6}$ and
$\widehat s_{j^*}\le m(w)+\frac{\eps n}{6}$. On this event, since
$m(\tilde w)\le \sum_{i\in[j^*]}\tilde w[i]$, we have
$\frac{n}{\alpha}m(\tilde w)\le \widehat s_{j^*}
\le m(w)+\frac{\eps n}{6}$. Therefore,
\begin{align*}
    \frac{n}{\alpha}\widetilde{\Delta}
    &=
    \frac{n}{\alpha}s(\tilde w)
    -2\frac{n}{\alpha}m(\tilde w)\\
    &\ge
    s(w)-\frac{\eps n}{6}
    -2\left(m(w)+\frac{\eps n}{6}\right)\\
    &=
    \Delta(w)-\frac{\eps n}{2}\\
    &\ge
    \frac{\eps n}{2}.
\end{align*}
Thus the tester rejects with probability at least $\frac{11}{12}$. Accounting
for the query cap still gives rejection probability at least $\frac{2}{3}$.
\end{proof}

\begin{lemma}[Completeness]\label{lem:excursion-completeness}
If $w\in L_{\ell,r}$, then
$\Pr[\text{\Cref{alg:excursion-tester} accepts } w]\ge 2/3$.
\end{lemma}

\begin{proof}
Assume $w\in L_{\ell,r}$. Then $s(w)=0$ and $s_j(w)\ge 0$ for all
$j\in\{0,\ldots,n\}$. In this case, the sampled walk can be written as its
expected value plus sampling noise. The expected value is a scaled copy of a
walk that stays nonnegative and ends at $0$, so only the sampling noise can
create a negative dip or a positive final sum.

Let $X_i$ be the indicator of the event that index $i$ is sampled, so
$\Pr[X_i=1]=p=\frac{\alpha}{n}$. For each $j$, define
\[
    Z_j=\sum_{i\in[j]}w[i](X_i-p).
\]
The random variables $w[i](X_i-p)$ are independent and have mean zero. Moreover,
\[
    s_j(\tilde w)
    =
    \sum_{i\in[j]}w[i]X_i
    =
    p s_j(w)+Z_j.
\]
Since $s_j(w)\ge 0$ for all $j$, we have
$m(\tilde w)\ge \min_{0\le j\le n}Z_j$. Also, since $s(w)=0$, we have
$s(\tilde w)=Z_n$.

We next bound the variance of $Z_n$. Since $|w[i]|\le \rho$,
\[
    \Var(Z_n)
    =
    \sum_{i\in[n]}w[i]^2\Var(X_i)
    \le
    \rho^2 n p
    =
    \rho^2\alpha.
\]
Let $\tau=\sqrt{24\rho^2\alpha}$. By Kolmogorov's inequality (\Cref{thm:kolmogorov-ineq}),
\[
    \Pr\left[\min_{0\le j\le n}Z_j\le -\tau\right]
    \le
    \frac{\Var(Z_n)}{\tau^2}
    \le
    \frac{1}{24}.
\]
By Chebyshev's inequality,
\[
    \Pr[Z_n\ge \tau]
    \le
    \frac{\Var(Z_n)}{\tau^2}
    \le
    \frac{1}{24}.
\]
Thus, with probability at least $\frac{11}{12}$, both
$m(\tilde w)\ge -\tau$ and $s(\tilde w)\le \tau$. On this event,
$\widetilde{\Delta}=s(\tilde w)-2m(\tilde w)\le 3\tau$. Therefore,
\[
    \frac{n}{\alpha}\widetilde{\Delta}
    \le
    \frac{3n\tau}{\alpha}
    =
    3n\sqrt{\frac{24\rho^2}{\alpha}}.
\]
Since $\alpha=\frac{C\rho^2}{\eps^2}$, this quantity is less than
$\frac{\eps n}{2}$ for a sufficiently large constant $C$. Thus the tester
accepts with probability at least $\frac{11}{12}$. Accounting for the query
cap, the acceptance probability is at least $\frac{2}{3}$.
\end{proof}

\begin{proof}[Proof of \Cref{thm:excursion-up}]
For $n\le \frac{C\rho^2}{\eps^2}$, the tester queries the entire input and
decides membership exactly. Otherwise, \Cref{alg:excursion-tester} uses
$O(\alpha)=O(\rho^2/\eps^2)$ queries after applying the query cap described
above. Since $\ell,r$ are fixed, this is $O_{\ell,r}(1/\eps^2)$. Completeness follows from \Cref{lem:excursion-completeness}, and soundness
follows from \Cref{lem:excursion-soundness}. The tester is nonadaptive, by definition.
\end{proof}

\subsection{An optimal tester for the Dyck-1 language}\label{sec:D1-tester}

In this section, we leverage our $L_{\ell,r}$ tester to get an optimal tester for $D_1$, the language of balanced parentheses. 
Identify $0^L$ with $+1$ and $0^R$ with $-1$; partial sums are then the heights of the walk, and $D_1$ is exactly $L_{1,1} \cap \{-1,1\}^*$.
We show that our $L_{1,1}$ tester yields a tester for $D_1$.

\begin{corollary}\label{cor:d1-tester}
    For all $\eps \in(0, 1)$, language $D_1$ has an $\eps$-tester that makes $O(1/\eps^2)$ queries. 
\end{corollary}

\begin{proof}
    Let $\mathcal{A}$ be the algorithm which, on input $w \in \{-1,1\}^*$, immediately rejects if $|w|$ is odd and runs our $\eps$-tester for $L_{1,1}$ (\Cref{alg:excursion-tester}), otherwise. This algorithm inherits the query complexity $O(1/\eps^2)$ of the $L_{1,1}$ tester. 
    
    It remains to analyze correctness. When the length of $w$ is odd, it cannot be in $D_1$, so the tester correctly rejects.
    Now assume $|w|$ is even.
    It suffices to show that $\dist(w,D_1) = \dist(w,L_{1,1})$.
    As $D_1 \subseteq L_{1,1}$, it follows that $\dist(w, D_1) \geq \dist(w, L_{1,1})$.
    To show $\dist(w,D_1) \leq \dist(w, L_{1,1})$, fix $w' \in L_{1,1}$ such that $\dist(w, w') = \dist(w, L_{1,1})$.
    As $w' \in L_{1,1}$, its final sum $s(w') = 0$.
    Also, as $w'$ has even length, it must contain an even number of 0s.
    Pair the 0s consecutively from left to right and replace each pair with a $1$ followed by a $-1$ to get a string $\hat{w} \in \{ -1, 1 \}^*$.
    The pairs are disjoint, so any prefix splits at most one pair, and only after the first element.
    Hence, $s_j(\hat{w}) - s_j(w') \in \{ 0, 1 \}$ for all $j$.
    So $s_j(\hat{w}) \ge 0$ and $s(\hat{w}) = 0$, and $\hat{w}$ has no 0s.
    Thus, $\hat{w} \in D_1$.
    Furthermore, $\dist(w,\hat{w}) \le \dist(w,w')$ since we only replace indices $i$ where $w'[i] = 0 \neq  w[i]$.
    So $\dist(w,D_1) \leq \dist(w, \hat{w}) \le \dist(w, w') = \dist(w, L_{1,1})$.
\end{proof}

\section{The optimality of our testers}\label{sec:optimality}
Finally, we show that our testers in \Cref{sec:excursion-languages} have optimal query complexity. Specifically, \Cref{sec:D1-lb} gives a lower bound for testing the Dyck-1 language, $D_1$, and \Cref{sec:excursion-languages-lb} gives a lower bound for testing all excursion languages.
These lower bounds match the dependence on $\eps$ in the complexity of our testers in \Cref{sec:excursion-languages}.

\subsection{A lower bound for the Dyck-1 language}\label{sec:D1-lb}
Now we prove the following theorem that shows that our tester for $D_1$ in \Cref{sec:D1-tester} is optimal.
\begin{theorem}\label{thm:D1-lb}
For every sufficiently small $\eps>0$, every (two-sided error, adaptive)
$\eps$-tester for $D_1$ has query complexity
$\Omega(1/\eps^2)$.
\end{theorem}

The proof of \Cref{thm:D1-lb} uses the following bound.
Let $\cH_0$ and $\cH_1$ be the uniform distributions over strings in $\{ -1, 1 \}^m$ whose symbols sum to $0$ and to $2 \gamma m$, respectively.
\cref{lem:distinguishing-lb} shows that distinguishing them requires $\Omega(1/\gamma^2)$ queries.
It follows from the standard fact that estimating the bias of a coin within $\pm\gamma$ with constant probability requires $\Omega(1/\gamma^2)$ coin tosses.
Recall that we write the partial sums of a word $w \in \{-1,1\}^n$ as $s_0(w)=0$, 
$s_j(w)=\sum_{i\in[j]} w[i]$ for all $j\in[n]$, and the final sum
$s(w)=s_n(w)$.

\begin{lemma}[Distinguishing lower bound]\label{lem:distinguishing-lb}
Let $\gamma\in(0,1/10)$, and let $m\in\bN$ be sufficiently large as a function of
$1/\gamma$. We say an algorithm $\cA$ {\em distinguishes} distributions $\cH_0$ and $\cH_1$ if
\[
    \left| \Pru{w \sim \cH_0}{\cA(w) = 1} - \Pru{w \sim \cH_1}{\cA(w) = 1} \right| 
    \geq \frac{1}{3}.
\]
Then every deterministic algorithm with query access to a string $w\in\{-1,1\}^m$  
requires $\Omega(1/\gamma^2)$ queries to distinguish between $w\sim\mathcal H_0$ and $w\sim\mathcal H_1$, where
\begin{align*}
    \mathcal H_0 &= \operatorname{Uniform}\{w\in\{-1,1\}^m : s(w) = 0\};\\
    \mathcal H_1 &= \operatorname{Uniform}\{w\in\{-1,1\}^m : s(w) = 2 \gamma m\}.
\end{align*}
\end{lemma}

\begin{proof}
    Fix a deterministic $q$-query algorithm and assume w.l.o.g.\ that it never repeats a query. By symmetry, the transcript distribution under $\mathcal H_b$ is the same as the transcript obtained by answering each new query by the next draw without replacement from an urn with the corresponding number of $1$'s and $-1$'s
    (i.e., with $m/2$ ones for $\mathcal{H}_0$ and $m/2 + \gamma m$ ones for $\mathcal{H}_1$).
    By the sampling-without-replacement versus sampling-with-replacement bound of Diaconis and Freedman \cite[Theorem 4]{DiaconisF80}, this distribution on the query answers is within $O(q/m)$ in total variation distance of the distribution obtained from $q$ i.i.d.\ Bernoulli samples of bias $1/2$ in the first case and $1/2+\gamma$ in the second case.
    
    It suffices to rule out algorithms with \(q \le c/\gamma^2\), for a sufficiently
    small constant \(c\). We choose \(m\) sufficiently large as a function of \(1/\gamma\)
    so that \(q/m=o(1)\) throughout this range.
    Every $q$-query distinguisher for $\cH_0$ and $\cH_1$ yields a
    $q$-sample distinguisher between a fair coin and a coin of bias
    $1/2+\gamma$, up to a negligible change in the success probability.
    By the standard lower bound for estimating the bias of a coin, e.g., as in the proof of Bar-Yossef~\cite[Theorem~8.4]{BarYossef03ECCC}, this requires $q=\Omega(1/\gamma^2)$.
\end{proof}

\begin{proof}[Proof of \Cref{thm:D1-lb}]
    Recall that $D_1$ is defined over the alphabet $\{-1,1\}$ where $1$ corresponds to open parentheses and $-1$ corresponds to closing parentheses.
    Fix $\eps \in(0,1/30)$, let $\gamma = 3\eps$, and $m$ be large enough for \cref{lem:distinguishing-lb}; inputs have length $n = 3m$.
    Let $\cH_0$ and $\cH_1$ be the distributions over $\{ -1, 1 \}^m$ from \cref{lem:distinguishing-lb}.
    Let
    \[
        \phi(x) = \term{1}^m \circ  x \circ \term{-1}^m.
    \]
    We apply Yao's principle (\cref{clm:yao}) on distributions $\cD_0$ and $\cD_1$, where $\cD_b$ samples $x \sim \cH_b$ and outputs $\phi(x)$.
    By \cref{lem:distinguishing-lb}, we know that distinguishing $\cD_0$ and $\cD_1$ will take $\Omega(1/\gamma^2) = \Omega(1/\eps^2)$ queries, so all we need to do is verify that the distributions are supported correctly.
    If $s(x)=0$, then $\phi(x)\in D_1$. Indeed, the initial block $\term{1}^m$
    ensures that the partial sum never becomes negative while reading $x$, and
    the final block $\term{-1}^m$ ensures $s(\phi(x)) = 0$. 
    On the other hand, suppose that $s(x) = 2\gamma m.$ Observe that $s(\phi(x)) = 2\gamma m$ as well. Changing one symbol decreases $s(w)$ by at most $2$, so at least $\gamma m$ symbols must differ between $\phi(x)$ and every $w' \in D_1$ as $s(w') = 0$.
    Thus, $\dist(\phi(x),D_1) \geq \gamma m =3\eps m$. Since $|\phi(x)|=3m$, we get $\phi(x)$ is $\eps$-far from $D_1$.
\end{proof}

\subsection{A lower bound for all excursion languages}\label{sec:excursion-languages-lb}
 
Finally, we show that 
our tester for $L_{\ell,r}$ in \Cref{sec:excursion-languages} has optimal dependence on $\eps$.
\begin{theorem}\label{thm:Llr-lb}
Fix $\ell,r\in\bN$. For every sufficiently small $\eps>0$, every (two-sided error, adaptive)
$\eps$-tester for $L_{\ell,r}$ has query complexity
$\Omega(1/\eps^2)$.
\end{theorem}
\begin{proof}[Proof of \Cref{thm:Llr-lb}]
    Let $\kappa=\min\{\ell,r\}$. We reduce from the distinguishing problem in
    \Cref{lem:distinguishing-lb}. Fix $\eps\in(0,1/30)$ and let
    $\gamma=3\eps$.
    
    {\bf Case 1:} $r\ge \ell$, i.e.,  $\kappa=\ell$. Let $\mathcal A$ be an
    $\eps$-tester for $L_{\ell,r}$. Given 
    $x\in\{-1,1\}^m$, define
    \[
        \phi(x)=\term{\kappa}^m \circ \kappa x \circ \term{-\kappa}^m,
    \]
    where $\kappa x$ denotes the word obtained from $x$ by multiplying every
    symbol by $\kappa$.
    We apply Yao's principle (\cref{clm:yao}) on distributions $\cD_0$ and $\cD_1$, where $\cD_b$ samples $x \sim \cH_b$ and outputs $\phi(x)$.
    By \cref{lem:distinguishing-lb}, we know that distinguishing $\cD_0$ and $\cD_1$ will take $\Omega(1/\gamma^2) = \Omega(1/\eps^2)$ queries, so all we need to do is verify that the distributions are supported correctly.
    
    If $s(x)=0$, then $\phi(x)\in L_{\ell,r}$. Indeed, the first block raises
    the walk to height $\kappa m$, the partial sums inside the middle block are nonnegative (because it has at most $m$ negative symbols), and the final block brings the walk back to $0$. On the other hand, suppose $s(x)=2\gamma m$. Then $s(\phi(x))=2\gamma\kappa m$. Since $\kappa=\ell$, changing one symbol of
    $\phi(x)$ decreases the final sum by at most $2\kappa$: a symbol
    $\kappa$ can be changed to $-\kappa=-\ell$, while a $-\kappa$
    cannot be decreased. Hence, each word in $L_{\ell,r}$ differs from
    $\phi(x)$ on at least $\gamma m$ indices and thus, $\dist(\phi(x),L_{\ell,r}) \ge \gamma m = 3\eps m$. As $|\phi(x)|=3m$, we have that $\phi(x)$ is $\eps$-far from $L_{\ell, r}$.
    
    {\bf Case 2:} $\ell>r$, i.e., $\kappa=r$. This case is identical, except
    that  we use the symmetric version of
    \Cref{lem:distinguishing-lb}, where $\mathcal H_1$ is uniform over
    strings $x\in\{-1,1\}^m$ satisfying $s(x)=-2\gamma m$. This version
    follows from \Cref{lem:distinguishing-lb} by negating all symbols. The
    same map
    \[
        \phi(x)=\term{\kappa}^m \circ \kappa x \circ \term{-\kappa}^m
    \]
    maps $\mathcal H_0$ to words in $L_{\ell,r}$. If
    $s(x)=-2\gamma m$, then $s(\phi(x))=-2\gamma\kappa m$. Since
    $\kappa=r$, changing one symbol of $\phi(x)$ can increase the final sum
    by at most $2\kappa$: a symbol $-\kappa$ can be changed to
    $\kappa=r$, while a symbol $\kappa$ cannot be increased. Thus, again
    $\phi(x)$ is $\eps$-far from $L_{\ell,r}$, and the same reduction gives
    the lower bound $\Omega(1/\eps^2)$.
\end{proof}

\section*{Acknowledgements}

An inspiration for this work was a project done by Ludmila Glinskih and Tim Jackman for a course on Sublinear Algorithms taught by Sofya Raskhodnikova \cite{Raskhodnikova25}.

\bibliography{bib} 
\bibliographystyle{alpha}

\crefalias{section}{appendix}
\crefalias{subsection}{appendix}

\appendix
\section{Yao's principle}\label{sec:yao}

We use the following form of Yao's minimax principle \cite{Yao77}.
 
\begin{claim}[Yao's principle]\label{clm:yao}
    Let $L$ be a language over $\Sigma$, let $\varepsilon\in(0,1)$, and let $N,q\in\mathbb{N}$.
    Let $\mathcal{D}^+$ be a distribution over $L\cap\Sigma^N$, let $\mathcal{D}^-$ be a distribution over $\Sigma^N$, and let $\eta$ be the probability that $y\sim\mathcal{D}^-$ is not $\varepsilon$-far from~$L$.
    Suppose that every deterministic algorithm $\mathcal{A}$ that makes at most $q$ queries satisfies
    \[
        \Big|\Pr_{x\sim\mathcal{D}^+}[\mathcal{A}(x)\text{ accepts}]-\Pr_{y\sim\mathcal{D}^-}[\mathcal{A}(y)\text{ accepts}]\Big|<\frac13-\eta.
    \]
    Then every $\varepsilon$-tester for $L$ makes more than $q$ queries on some input of length $N$.
    The same holds if the algorithms $\mathcal{A}$ and the tester are all nonadaptive.
\end{claim}
\begin{proof}
    Suppose that $T$ is an $\varepsilon$-tester for $L$ that makes at most $q$ queries on every input of length $N$, and let $R$ denote its random coins.
    For every fixed value $r$ of $R$, the algorithm $T_r$ that runs $T$ with coins $r$ is deterministic, makes at most $q$ queries, and is nonadaptive if $T$ is.
    Every $x$ in the support of $\mathcal{D}^+$ is in $L$, so $\Pr_{x\sim\mathcal{D}^+,R}[T_R(x)\text{ accepts}]\ge 2/3$.
    Every $y$ that is $\varepsilon$-far from $L$ is accepted with probability at most $1/3$, so $\Pr_{y\sim\mathcal{D}^-,R}[T_R(y)\text{ accepts}]\le\eta+1/3$.
    Hence
    \[
      \mathbb{E}_R\Big[\Pr_{x\sim\mathcal{D}^+}[T_R(x)\text{ accepts}]-\Pr_{y\sim\mathcal{D}^-}[T_R(y)\text{ accepts}]\Big]\ge\frac13-\eta,
    \]
    so some fixed value $r$ satisfies the same inequality, and $T_r$ contradicts the assumption.
\end{proof}



Suppose there exists an event $\Bad$ that helps the algorithm distinguish between $\Dyes$ and $\Dno$ and occurs with the same probability under both distributions. Then standard arguments yield the following claim. Recall from \cref{sec:HS-lb} that $\Tru{\cA}{\cD}$ represents the distribution of transcripts of algorithm $\cA$ running on input drawn from $\cD$.

\begin{claim}\label{claim:indistinguishability-via-probability-of-bad}
    Let $\Dyes$ and $\Dno$ be defined as in \cref{clm:yao}. Let $\Bad$ be the event that occurs with the same probability under both distributions, such that $\Tru{\cA}{\Dyes \mid \overline{\Bad}} = \Tru{\cA}{\Dno \mid \overline{\Bad}}$. Then
    \begin{align*}
         &\left| \Pru{x \sim \Dyes}{\cA(x) \text{ accepts}} - \Pru{y \sim \Dno}{\cA(y) \text{ accepts}} \right|
         \leq \Pr[\Bad].
    \end{align*}        
\end{claim}
    
\begin{proof} By the law of total probability and using, in the last equality, the fact that $\Tru{\cA}{\Dyes \mid \overline{\Bad}} = \Tru{\cA}{\Dno \mid \overline{\Bad}}$, we get
        \begin{align*}
         &\left| \Pru{x \sim \Dyes}{\cA(x) \text{ accepts}} - \Pru{y \sim \Dno}{\cA(y) \text{ accepts}} \right|\\
         &= \Big| \Pru{x \sim \Dyes \mid \overline{\Bad}}{\cA(x) \text{ accepts}} \cdot \Pr[\overline{\Bad}] + \Pru{x \sim \Dyes \mid \Bad}{\cA(x) \text{ accepts}} \cdot \Pr[\Bad]\\
         &\qquad - \Pru{y \sim \Dno \mid \overline{\Bad}}{\cA(y) \text{ accepts}} \cdot \Pr[\overline{\Bad}] - \Pru{y \sim \Dno \mid \Bad}{\cA(y) \text{ accepts}} \cdot \Pr[\Bad] \Big|\\
         &\le \Pr[\overline{\Bad}] \cdot \left| \Pru{x \sim \Dyes \mid \overline{\Bad}}{\cA(x) \text{ accepts}} - \Pru{y \sim \Dno \mid \overline{\Bad}}{\cA(y) \text{ accepts}} \right|\\
         &\qquad+ \Pr[\Bad] \cdot \left| \Pru{x \sim \Dyes \mid \Bad}{\cA(x) \text{ accepts}} - \Pru{y \sim \Dno \mid \Bad}{\cA(y) \text{ accepts}} \right|\\
         &= 0 + \Pr[\Bad] \cdot \left| \Pru{x \sim \Dyes \mid \Bad}{\cA(x) \text{ accepts}} - \Pru{y \sim \Dno \mid \Bad}{\cA(y) \text{ accepts}} \right| &\\
         &\le \Pr[\Bad]. &\qedhere
    \end{align*}
\end{proof}

\section{Inequalities}\label{sec:bounds}

We use the following tail bounds in our proof of \cref{thm:excursion-up} in \Cref{sec:excursion-languages}.

\begin{theorem}[Kolmogorov's Inequality]\label{thm:kolmogorov-ineq}
Let \(X_1,\ldots,X_n\) be independent mean-zero random variables
with finite variances.
Let $S_j = \sum_{i=1}^j X_i$ for all $j \in [n]$. Then for all $t > 0$,
\[
    \Pr\left[\max_{j \in [n]}\left|S_j\right| \ge t\right] \leq \frac{\Var[S_n]}{t^2} .
\]
\end{theorem}

\begin{theorem}[Bernstein's Inequality]\label{thm:bernstein-ineq}
    Let $X_1, \dots, X_n$ be independent random variables satisfying $|X_i - \Ex{X_i}| \leq \mu$ for all $i \in [n]$.
    Let $X = \sum_{i \in [n]} X_i$.
    Let $\sigma^2 = \Var(X)$.
    Then for all $t > 0$,
    \[
        \Pr\left[\left| 
        X- \Ex{X} \right| \ge t\right] \le 2 \exp \left( -\frac{t^2}{2\sigma^2 + \frac{2}{3} \mu t} \right).
    \]
\end{theorem}

We also use an anti-concentration bound (\cref{lem:block-anticoncentration}), which we now prove with $C_0 = 5$ for all $b \ge 15$.
The idea is to split each uniform random variable into a fair coin with a large weight and a small remainder.
The coins spread the sum over $\Omega(\sqrt{s})$ multiples of the weight, and the remainders spread it evenly over the residues modulo the weight.
We first handle variables that take an even number of values, and then reduce the general case to it by
conditioning on the variables that take the value $b$, which are rare.
We use the following bound on binomial point probabilities.
A version of this bound for large $s$ appears as \cite[Lemma~5.11]{FischerMS18};
since we need it for all $s$, we include a short proof. 
 
\begin{lemma}\label{lem:binomial-point}
    For all $s \in \mathbb{N}$ and $t \in \mathbb{Z}$, if random variables $Y_1, \ldots, Y_s$ are independent and uniform over $\{0, 1\}$, then $\Pr \left[ \sum_{i=1}^{s} Y_i = t \right] \le \frac 1{\sqrt{s}}$.
\end{lemma}
 
\begin{proof}
    The probability equals $\binom{s}{t} 2^{-s}$, which is maximized at $t = \lfloor s/2 \rfloor$.
    For all integers $r \ge 0$, let $P_r := \binom{2r}{r} 4^{-r}$.
    Since $(2r)! = \prod_{i=1}^{r} (2i-1)(2i)$ and $2^r r! = \prod_{i=1}^{r} 2i$, we have $P_r = \prod_{i=1}^{r} \frac{2i-1}{2i}$.
    Moreover, $\frac{2i-1}{2i} \le \frac{2i}{2i+1}$ for all $i \ge 1$, because $(2i-1)(2i+1) \le (2i)^2$.
    Hence,
    \[
        P_r^2 \le \prod_{i=1}^{r} \frac{2i-1}{2i} \cdot \frac{2i}{2i+1}
        = \prod_{i=1}^{r} \frac{2i-1}{2i+1} = \frac{1}{2r+1}.
    \]
    If $s = 2r$, the maximum probability is $P_r \le 1/\sqrt{2r+1} < 1/\sqrt{s}$.
    If $s = 2r+1$, then $\binom{2r+1}{r} = \binom{2r}{r} + \binom{2r}{r-1} \le 2\binom{2r}{r}$, so the maximum probability is $\binom{2r+1}{r} 2^{-(2r+1)} \le P_r \le 1/\sqrt{2r+1} = 1/\sqrt{s}$.
\end{proof}
 
\begin{claim}\label{clm:even-support}
    For all $s, m \in \mathbb{N}$ and $t \in \mathbb{Z}$, if random variables $Z_1, \ldots, Z_s$ are independent and uniform over $\{0, 1, \ldots, 2m-1\}$, then
    $\Pr \left[ \sum_{i=1}^{s} Z_i = t \right] \le \frac 1{m\sqrt{s}}$.
\end{claim}
 
\begin{proof}
    For each $i \in [s]$, write $Z_i = mY_i + R_i$, where $Y_i := \lfloor Z_i/m \rfloor \in \{0, 1\}$ is the \emph{coin} and $R_i := Z_i - mY_i \in \{0, \ldots, m-1\}$ is the \emph{remainder}.
    The map $(y, r) \mapsto my + r$ is a bijection from $\{0,1\} \times \{0, \ldots, m-1\}$ to $\{0, \ldots, 2m-1\}$.
    Hence, the pair $(Y_i, R_i)$ is uniform over $\{0,1\} \times \{0, \ldots, m-1\}$; that is, $Y_i$ and $R_i$ are independent, $Y_i$ is uniform over $\{0,1\}$, and $R_i$ is uniform over $\{0, \ldots, m-1\}$.
    Let $Y := \sum_{i=1}^{s} Y_i$ and $R := \sum_{i=1}^{s} R_i$.
    Therefore,
    \begin{align*}
        \Pr \left[ \sum_{i=1}^{s} Z_i = t \right]
        &= \sum_{y \in \mathbb{Z}} \Pr[Y = y] \cdot \Pr[R = t - my]\\
        &\le \frac{1}{\sqrt{s}} \sum_{y \in \mathbb{Z}} \Pr[R = t - my]
        = \frac{1}{\sqrt{s}} \Pr[R \equiv t\ (\mathrm{mod}\ m)].
    \end{align*}
    The inequality holds by \cref{lem:binomial-point}.
    The last equality holds because the events $\{R = t - my\}$ for $y \in \mathbb{Z}$ are disjoint, and their union is the event $\{R \equiv t\ (\mathrm{mod}\ m)\}$.
    Finally, fix the values of $R_2, \ldots, R_s$.
    Exactly one value of $R_1 \in \{0, \ldots, m-1\}$ satisfies $R_1 \equiv t - R_2 - \cdots - R_s\ (\mathrm{mod}\ m)$, so $\Pr[R \equiv t\ (\mathrm{mod}\ m)] = 1/m$.
\end{proof}
 
\begin{proof}[Proof of \cref{lem:block-anticoncentration}]
    We show that the lemma holds with $C_0 = 5$ for all $b \ge 15$.
    Fix $s \in \mathbb{N}$, an integer $b \ge 15$, and $t \in \mathbb{Z}$, and let $W := \sum_{i=1}^{s} W_i$.
    Let $m := \lfloor (b+1)/2 \rfloor$, so that $m \ge b/2$ and $2m \in \{b, b+1\}$.
    Thus, a variable $W_i$ can lie outside $\{0, \ldots, 2m-1\}$ only if $b$ is even and $W_i = b$.
    
    Let $D := \{i \in [s] : W_i = 2m\}$, and let $N := |D|$.
    The events $\{i \in D\}$ for $i \in [s]$ are independent, and each has probability $p := (b+1-2m)/(b+1) \le 1/(b+1)$. 
    Conditioned on $D$, the variables $W_i$ with $i \notin D$ are independent and uniform over $\{0, \ldots, 2m-1\}$, and $W$ equals $bN$ plus their sum.
    If $N \le s/2$, then there are $s - N \ge s/2$ such variables, so \cref{clm:even-support} gives
    \[
        \Pr[W = t \mid D] \le \frac{1}{m\sqrt{s-N}}
        \le \frac{2\sqrt{2}}{b\sqrt{s}}.
    \]
    It remains to bound $\Pr[N > s/2]$.
    Let $R := \lfloor s/2 \rfloor + 1$, so that $N > s/2$ if and only if $N \ge R$, and $R \ge (s+1)/2$.
    By a union bound over the $\binom{s}{R} \le 2^{s}$ sets of $R$ indices in $[s]$,
    \[
        \Pr[N > s/2] \le \binom{s}{R} p^{R}
        \le 2^{s} p^{(s+1)/2} = 2p\,(4p)^{(s-1)/2}
        \le 2p \cdot 2^{-(s-1)} \le \frac{2p}{\sqrt{s}} < \frac{2}{b\sqrt{s}}.
    \]
    Here, we used that $4p \le 4/(b+1) \le 1/4$, since $b \ge 15$, and that
    $2^{s} \ge 2\sqrt{s}$ for all $s \ge 1$. Combining the two bounds,
    \[
        \Pr[W = t]
        \le \Pr[N > s/2] + \max_{D \,:\, |D| \le s/2} \Pr[W = t \mid D]
        < \frac{2}{b\sqrt{s}} + \frac{2\sqrt{2}}{b\sqrt{s}}
        < \frac{5}{b\sqrt{s}}. \qedhere
    \]
\end{proof}

\section{Testability and relationships of subclasses of CFLs}\label{sec:classes}

In this section, we situate our results within the broader context of property testing context-free languages (CFLs). In \Cref{sec:hierarchy}, we describe a hierarchy of linear languages, within which \HSD from \Cref{thm:HSD-lb} sits at the most restricted level. In \Cref{sec:complexity-diagram}, we survey known complexity bounds for subclasses of CFLs.

\subsection{The Nondeterministic Biautomata Hierarchy}\label{sec:hierarchy}
\begin{figure}
    \centering
    \includegraphics[]{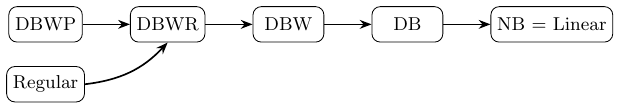}
    \caption{The hierarchy of language classes recognized by nondeterministic biautomata. An arrow from one class to another represents that the former class is contained in the latter.}
    \label{fig:NB-hierarchy}
\end{figure}

Just as the regular languages can be characterized via finite automata, the linear languages can be characterized using nondeterministic biautomata, introduced by \cite{HolzerJ14}. A biautomaton has two reading heads, one at each end of the input, and at each step the machine nondeterministically chooses which head to advance, reading \emph{either} the left \emph{or} the right end of the string. The set of languages recognized by these machines, NB, is exactly the linear languages. Requiring determinism, i.e., allowing exactly one computation path on every input, yields a strict subclass DB which contains the regular languages. Further restrictions on allowed transitions produce a hierarchy of subclasses of the linear languages \cite{JiraskovaK22a}. Our hard languages live at the bottom of the hierarchy, in a class that is incomparable to the regular languages (i.e., neither contains the other).
  
We describe these subclasses below in increasing order of restrictiveness. Since the transition conditions are technical, we instead give the type of grammars that generate its languages. The overall structure of the hierarchy can be seen in \cref{fig:NB-hierarchy}.
    \begin{description}[itemindent=-1em]
        \item[DBW] (deterministic biautomata \emph{weak from the right}): the languages generated by deterministic linear grammars as defined in \cite{HigueraO02}.
        \item[DBWR] (deterministic biautomata weak from the right \emph{with restricted final states}): the languages generated by \emph{linear} LL(1) grammars. LL(1) grammars are grammars with simple and efficient parsers which read left-to-right while looking at most one symbol ahead.
        \item[DBWP] (deterministic biautomata weak from the right \emph{with a passive final state}): exactly the Nasu--Honda deterministic linear languages from \cite{NasuH69}. Languages in this class are generated by grammars $G=(V,\Sigma,R,S)$ whose productions satisfy the following two conditions (where $V$ is the set of non-terminals, $\Sigma$ is the set of terminals, $R$ is the set of production rules, and $S$ is the start symbol).
            \begin{enumerate}
                \item Every production is of the form $A\to aBu$ or $A\to a$, where
                $a\in\Sigma$, $B\in V$, and $u\in\Sigma^*$.
                \item For all $A\in V$, $a\in\Sigma$, and
                $\alpha,\beta\in V\Sigma^*\cup\{\lambda\}$, if 
                $A\to a\alpha$ and $A\to a\beta$ are in $R$, then~$\alpha=\beta$.
            \end{enumerate}
    \end{description}

\subsection{Known query complexity bounds for subclasses of CFLs}\label{sec:complexity-diagram}

Property testing for formal languages has been studied across a number of language classes, with query complexities spanning from constant (depending only on $\eps$) to linear in the input length $n$. Alon et al.~\cite{AlonKNS00} showed that the regular languages sit as a testable core inside the CFLs: every regular language has a tester whose query complexity is independent of the input length and depends only on the distance parameter. They also showed that this testability does not extend to the wider class of CFLs by exhibiting a CFL which requires $\Omega(\sqrt{n})$ queries to test. Subclasses of CFL have also been studied and proven to not be testable with a constant number of queries. Deterministic one-counter languages (detOCL) were shown to contain languages requiring $\Omega(\log \log n)$ queries \cite{LachishNS08}, and the Dyck languages $(\{D_m\}_{m \in \bN})$ contain languages requiring $\Omega(\qb)$ queries (\Cref{thm:dyck-lb}). In \Cref{thm:HSD-lb}, we show that the class of linear languages contains a language which requires $\Omega(\qb)$ queries to test, and that this language sits in a class DBWP, at the bottom of the nondeterministic biautomaton hierarchy described in \Cref{sec:hierarchy}. This result shows that even very restricted classes of CFLs may not be {\em testable}.

However, some nonregular CFLs are easily testable.
In the same spirit as \cite{AlonKNS00}, Goldhirsh and Viderman \cite{GoldhirshV13} identified a restriction of the deterministic one-counter languages -- the weak deterministic one-counter languages -- that recovers constant-query testability, carving out a tractable subclass within an otherwise hard family. In \Cref{thm:excursion}, we show that excursion languages, a natural generalization of the Dyck-1 language, are testable. 

\Cref{fig:cfl-hierarchy} 
illustrates
the relationships between these classes, 
collecting the currently known bounds, both from prior work and from this paper. It mentions 3CNF properties as a non-context-free language class requiring $\Omega(n)$ queries to test \cite{Ben-SassonHR05}; whether any CFL requires $\omega(\sqrt n)$ queries remains open.

\begin{figure}[h]
    \centering
    \scalebox{0.75}{
    \begin{tikzpicture}[
    node distance = 4mm and 30mm,   
    every node/.style = {font=\Large},
    cls/.style       = {draw, rounded corners, inner sep=5pt, minimum width=12mm,
                         minimum height=6mm, align=center, fill=white},
    inc/.style       = {-{Stealth[length=2mm]}}
]

\node[cls,
      label={[font=\large, align=right, label distance=8pt]left:{%
              \textcolor{orange!70!black}{$\Omega(\sqrt{n})$}\\
              \textcolor{black}{\cite{AlonKNS00}}}}
] (CFG) {CFL};

\node[cls, below=0.25in of CFG,
      label={[font=\large, align=left, label distance=2pt]right:{%
              \textcolor{orange!70!black}{$\Omega(\log \log n)$}\\%
       \textcolor{black}{\cite{LachishNS08}}}}
] (detOCL) {detOCL};

\node[cls, right=1in of detOCL, yshift = -0.325in,
      label={[font=\large, align=left, label distance=2pt]right:{\textcolor{orange!70!black}{Adaptive: $\Omega(n^{2/5})$} \\ 
           \textcolor{orange!70!black}
           {Nonadaptive: $\Omega(n^{1/2})$} \\
              \textcolor{black}{(\Cref{thm:HSD-lb})}}}
] (Linear) {Linear};

\node[cls, left=0.325in of detOCL,
      label={[font=\large, align=right, label distance=2pt]left:%
             {\textcolor{blue!70!black}{$O(n^{2/5+\alpha})$} \cite{FischerMS18}\\%
              \textcolor{orange!70!black}{Adaptive: $\Omega(n^{2/5})$} (\Cref{thm:dyck-lb})
              \\%
              \textcolor{orange!70!black}{Nonadaptive: $\Omega(n^{1/2})$} (\Cref{thm:dyck-lb-nonadaptive})}}
] (Dm) {$\{D_m\}_{m \in \mathbb{N}}$};

\node[cls, below=0.25in of Dm,
    label={[font=\large, align=right, label distance=2pt]left:%
           {\textcolor{black}{$\Theta(1/\eps^{2})$}\\ \textcolor{black}{(\Cref{cor:d1-tester}, \Cref{thm:D1-lb})}}}
] (D1) {$\{D_1\}$};

\node[cls, below=0.5in of detOCL,
    label={[font=\large, align=left, label distance=2pt]right:%
           {\textcolor{blue!70!black}{$O(f(\eps))$} \\
           \textcolor{black}{\cite{GoldhirshV13}}}}] (Weak) {weak detOCL};

\node[cls, below=0.25in of D1,
      label={[font=\large, align=right, label distance=2pt]left:%
             {\textcolor{black}{$\Theta(1/\eps^{2})$}\\
             \textcolor{black}{(\Cref{thm:excursion})}}}] (Excursion) {excursions};

\node[cls, below=0.5in of Weak,
      label={[font=\large, align=right, label distance=2pt]left:{%
              \textcolor{blue!70!black}{$O(\log(\eps^{-1})/\eps)$}
              \textcolor{black}{\cite{BathieFM25}} \\
              \textcolor{orange!70!black}{1-sided error: $\Omega(\log(\eps^{-1})/\eps)$}
              \textcolor{black}{\cite{BathieFM25}}}}
] (Regular) {Regular};

\node[cls, right=1in of Regular,
    label={[font=\large, align=left, label distance=2pt]right:%
           {\textcolor{orange!70!black}{Adaptive: $\Omega(n^{2/5})$} \\ 
           \textcolor{orange!70!black}
           {Nonadaptive: $\Omega(n^{1/2})$} \\
           \textcolor{black}{(\Cref{thm:HSD-lb})}}}] (DBWP) {DBWP};

\node[cls, above = 0.575in of Linear,
      label={[font=\large, align=left, label distance=2pt]right:{%
              \textcolor{orange!70!black}{$\Omega(n)$ }\\
              \textcolor{black}{\cite{Ben-SassonHR05} }}}
] (3CNF) {3CNF};

\draw[inc] (Dm)        -- (CFG);
\draw[inc] (detOCL)    -- (CFG);
\draw[inc] (Linear)    to[bend right = 30] (CFG);
\draw[inc] (D1)        -- (Dm);
\draw[inc] (D1)        -- (detOCL);
\draw[inc] (Excursion) -- (detOCL);
\draw[inc] (Weak)      -- (detOCL);
\draw[inc] (DBWP)      -- (Linear);
\draw[inc] (Regular)   -- (Weak);
\draw[inc] (Regular) to[bend right=35] (Linear);

\end{tikzpicture}
    }
    \caption{Known query complexity bounds for language classes, with arrows indicating containment. For each class, big-O and $\Theta$ bounds apply to every language in the class, while lower bounds are witnessed by some language in the class.}
    \label{fig:cfl-hierarchy}
\end{figure}

\end{document}